\documentclass[USenglish,oneside,twocolumn]{article}

\usepackage[utf8]{inputenc}

\usepackage{amssymb}
\usepackage{bbm}
\usepackage{tikz}
\usepackage{amsmath}
\usepackage{amssymb}
\usepackage{subfigure}
\usepackage{balance}

\newtheorem{theorem}{\bf Theorem}
\newtheorem{proof}{\bf Proof}

\newcommand{\rv}[1]{\boldsymbol{#1}}
\newcommand{\pr}[1]{\textrm{Pr}\{#1\}}

\newcommand{\minimize}[1]{\underset{#1}{\operatorname{\mathbf{min}}}\,\,}
\newcommand{\maximize}[1]{\underset{#1}{\operatorname{\mathbf{max}}}\,\,}

\newcommand{\argmin}[1]{\underset{#1}{\operatorname{\mathbf{argmin}}}\,\,}
\newcommand{\st}{\mathbf{s.\,t.}\,\,}

\begin{document}

  \author{Reza Shokri\\University of Texas at Austin\\shokri@cs.utexas.edu}
  \date{}
  \title{Privacy Games: Optimal User-Centric Data Obfuscation}
\maketitle

  \begin{abstract}
{
Consider users who share their data (e.g., location) with an untrusted service provider to obtain a personalized (e.g., location-based) service.  Data obfuscation is a prevalent user-centric approach to protecting users' privacy in such systems: the untrusted entity only receives a noisy version of user's data.  Perturbing data before sharing it, however, comes at the price of the users' utility (service quality) experience which is an inseparable design factor of obfuscation mechanisms.  The entanglement of the utility loss and the privacy guarantee, in addition to the lack of a comprehensive notion of privacy, have led to the design of obfuscation mechanisms that are either suboptimal in terms of their utility loss, or ignore the user's information leakage in the past, or are limited to very specific notions of privacy which e.g., do not protect against adaptive inference attacks or the adversary with arbitrary background knowledge.\\[8pt]
In this paper, we design user-centric obfuscation mechanisms that impose the minimum utility loss for guaranteeing user's privacy.  We optimize utility subject to a joint guarantee of differential privacy (indistinguishability) and distortion privacy (inference error).  This double shield of protection limits the information leakage through obfuscation mechanism as well as the posterior inference.  We show that the privacy achieved through joint differential-distortion mechanisms against optimal attacks is as large as the maximum privacy that can be achieved by either of these mechanisms separately.  Their utility cost is also not larger than what either of the differential or distortion mechanisms imposes.  We model the optimization problem as a leader-follower game between the designer of obfuscation mechanism and the potential adversary, and design adaptive mechanisms that anticipate and protect against optimal inference algorithms.  Thus, the obfuscation mechanism is optimal against any inference algorithm.  
}
\end{abstract}

\section{Introduction}\label{sec:introduction}

Data obfuscation is a mechanism for hiding private data by using misleading, false, or ambiguous information with the intention of confusing an adversary \cite{brunton2011vernacular}. A data obfuscation mechanism acts as a noisy information channel between a user's private data (secret) and an untrusted observer \cite{chatzikokolakis2008anonymity}. The noisier this channel is, the higher the privacy of the user will be. We focus on {\em user-centric} mechanisms, in which each user independently perturbs her secret before releasing it. Note that we are not concerned with database privacy, but with the privacy issues of releasing a single sensitive data sample (which however could be continuously shared over time). For example, consider a mobile user who is concerned about the information leakage through her location-based queries. In this case, obfuscation is the process of randomizing true locations so that the location-based server only receives the user's perturbed locations. 

By using obfuscation mechanisms, the {\em privacy} of a user and her {\em utility} experience are at odds with each other, as the service that the user receives is a function of what she shares with the service provider. There are problems to be addressed here. One is how to design an obfuscation mechanism that protects privacy of the user and imposes a {\em minimum} utility cost. Another problem is how to {\em guarantee} the user's privacy, despite the lack of a single best metric for privacy.

Regarding utility optimization, we define utility loss of obfuscation as the degradation of the user's service-quality expectation due to sharing the noisy data instead of its true value. Regarding privacy protection, there are two major metrics proposed in the literature. {\em Differential} privacy limits the information leakage through observation. But, it does not reflect the absolute privacy level of the user, i.e., what actually is learned about the user's secret. So, user would not know how close the adversary's estimate will get to her secret if she releases the noisy data, despite being sure that the relative gain of observation for adversary is bounded.  {\em Distortion} privacy (inference error) metric overcomes this issue and measures the error of inferring user's secret from the observation. This requires assumption of a prior knowledge which enables us to quantify absolute privacy, but is not robust to adversaries with arbitrary knowledge. Thus, either of these metrics alone is incapable of capturing privacy as a whole.

The problem of optimizing the tradeoff between privacy and utility has already been discussed in the literature, but notably for differential privacy in the context of statistical databases \cite{BrennerN10, GhoshRS09, GhoshRS12, GupteS10, LiHRMM10}. Regarding user-centric obfuscation mechanisms, \cite{ShokriTTHL12} solves the problem of maximizing distortion privacy under a constraint on utility loss. The authors construct the optimal adaptive obfuscation mechanism as the user's best response to the adversary's optimal inference in a Bayesian zero-sum game. In the same context, \cite{BordenabeCP14} solves the opposite problem, i.e., optimizing utility but for differential privacy. In both papers, the authors construct the optimal solutions using linear programming. 

Differential and distortion metrics for privacy complement each other. The former is sensitive to the likelihood of observation given data. The latter is sensitive to the joint probability of observation and data. Thus, by guaranteeing both, we encompass all the defense that is theoretically possible. In this paper, we model and solve the optimal obfuscation mechanism that: (i) minimizes utility loss, (ii) satisfies differential privacy, and (iii) guarantees distortion privacy, given a public knowledge on prior leakage about the secrets. We measure the involved metrics based on separate distance functions defined on the set of secrets. We model prior leakage as a probability distribution over secrets, that can be estimated from the user's previously released data. Ignoring such information leads to overestimating the user's privacy and thus designing a weak obfuscation mechanism (against adversaries who include such exposed information in their inference attack).\footnote{Note that the prior leakage is {\em not} equivalent to the adversary's knowledge. An adversary might have access to some information about the user's data through channels where the user is unaware of and has no control over. No protection mechanism can guarantee the distortion privacy against adversaries with arbitrary knowledge. Imagine the worst case where adversary knows the exact secret but through channels other than observation of the obfuscation mechanism. Therefore, our focus is on the user, and we incorporate what the user thinks has been leaked so far.}

A protection mechanism for distortion privacy metric can be designed such that it is optimal against a particular inference algorithm (e.g., Bayesian inference \cite{Berger85, Mackay03} as privacy attacks \cite{ShokriTLH11, TroncosoD09}). But, by doing so, it is not guaranteed that the promised privacy level can be achieved in practice: an adversarial observer can run inference attacks that are optimally tailored against the very obfuscation mechanism used by the user (regardless of the algorithm that the user assumes a priori). In fact, the adversary has the upper hand as he infers the user's secret (private information) {\em after} observing the output of the obfuscation mechanism. Thus, the obfuscation mechanisms must {\em anticipate} the adaptive inference attack that will follow the observation.  This enables us to design an obfuscation mechanism that is independent of the adversary's inference algorithm.

To address this concern, we adapt a game-theoretic notion of privacy for designing optimal obfuscation mechanisms against adaptive inference.  We formulate this game as a Stackelberg game and solve it using linear programming.\footnote{As opposed to \cite{ShokriTTHL12}, the game is not zero-sum anymore given that here user maximizes utility and observer minimizes privacy.}  We then add the differential privacy guarantee as a constraint in the linear program and solve it to construct the optimal mechanism. The result of using such obfuscation mechanism is that, not only the perturbed data samples are indistinguishable from the true secret (due to differential privacy bound), but also they cannot be used to accurately infer the secret using the prior leakage (due to distortion privacy measure). To the best of our knowledge, this work is the first to construct utility maximizing obfuscation mechanisms with such formal privacy guarantees. 

We illustrate the application of optimal protection mechanisms on a real data set of users' locations, where users want to protect their location privacy against location-based services. We evaluate the effects of privacy guarantees on utility cost. We also analyze the robustness of our optimal obfuscation mechanism against inference attacks with different algorithms and background knowledge. We show that our joint differential-distortion mechanisms are robust against adversaries with optimal attack and background knowledge. Moreover, the utility loss is at most equal to the utility loss of differential or distortion privacy, separately.

The novelty of this paper in the context of user-centric obfuscation is twofold:
\begin{itemize}
  \item We construct optimal obfuscation mechanisms that provably limit the user's privacy risk (i.e., by guaranteeing the user's distortion privacy) against {\em any} inference attack, with minimum utility cost.
  \item We design obfuscation mechanisms that optimally balance the tradeoff between utility and joint distortion-differential privacy. The solution is robust against adversary with arbitrary knowledge, yet it guarantees a required privacy given the user's estimation of the prior information leakage. 
\end{itemize}

\section{Related Work}\label{sec:relatedwork}

This paper contributes to the broad area of research that concerns designing obfuscation mechanisms, e.g., in the context of quantitative information flow \cite{MardzielAHC14, KopfB07, AlvimCPS12}, quantitative privacy in data sharing systems \cite{AndresBCP13, ShokriTTHL12, ShokriTLH11}, as well as differential privacy \cite{Dwork06, GhoshRS12, GupteS10, LiHRMM10}. The conflict between privacy and utility has been discussed in the literature \cite{BrickellS08, ioannidis2014privacy}. We build upon prevalent notions of privacy and protect it with respect to information leakage through both observation (differential privacy) and posterior inference (distortion privacy) while optimizing the tradeoff between utility and privacy. We also formalize this problem and solve it for user-centric obfuscation mechanisms, where it's each individual user who perturbs her secret data before sharing it with external observers (e.g., service providers).

The problem of perturbing data for differential and distortion privacy, separately, and optimizing their effect on utility has already been discussed in the literature. Original metric for differential privacy measures privacy of output perturbation methods in statistical databases \cite{Dwork06}. Assuming two statistical databases to be neighbor if they differ only in one entry, \cite{GhoshRS09} and \cite{GhoshRS12} design utility maximizing perturbation mechanisms for the case of counting queries. In \cite{GupteS10, LiHRMM10}, authors propose different approaches to designing perturbation mechanisms for counting queries under differential privacy. However, \cite{BrennerN10} presents some impossibility results of extending these approaches to other types of database queries. Under some assumptions about the utility metric, \cite{geng2012optimal} shows that the optimal perturbation probability distribution has a symmetric staircase-shaped probability density function. \cite{BartheKOZ12, ChatzikokolakisABP13, ReedP10} extend differential privacy metric using generic distance functions on the set of secrets. Some extensions of differential privacy also consider the problem of incorporating the prior knowledge into its privacy definition \cite{kifer2011no, he2014blowfish}. 

The most related paper to our framework, in this domain, is \cite{BordenabeCP14} where the authors construct utility-maximizing differentially private obfuscation mechanisms using linear programming. The authors prove an interesting relation between utility-maximizing differential privacy and distortion-privacy-maximizing mechanisms that bound utility, when distance functions used in utility and privacy metrics are the same. This, however, cannot guarantee distortion privacy for general metrics. The optimal differentially private mechanisms, in general, do not incorporate the available knowledge about the secret while achieving differential privacy.

Distortion privacy, which evaluates privacy as the inference error \cite{ShokriTLH11}, is a follow-up of information-theoretic metrics for anonymity and information leakage \cite{chatzikokolakis2008anonymity, diaz2003towards, KopfB07, serjantov2003towards}. This class of metrics is concerned with what can be inferred about the true secret of the user by combining the observation (of obfuscated information) and prior knowledge. The problem of maximizing privacy under utility constraint, assuming a prior, is proven to be equivalent to the user's best strategy in a zero-sum game against adaptive adversaries \cite{ShokriTTHL12}. With this approach, one can find the optimal strategies using linear programming. In fact, linear programming is the most efficient solution for this problem \cite{ConitzerT06}. However, if we want to guarantee a certain level of privacy for the user and maximize her utility, the problem cannot be modeled as a zero-sum game anymore and there has been no solution for it so far. We formalize this game, and construct a linear programming solution for these privacy games too.

Regarding the utility metric, we consider the expected distance between the observation and the secret as the utility metric \cite{BrennerN10, ChatzikokolakisABP13, GhoshRS09, ShokriTTHL12}. The distance function can depend on the user and also the application.

In the case of applying obfuscation over time, we need to update the user's estimation of the prior leakage according to what has been shared by the user~\cite{theodorakopoulos2014prolonging, danezis2013you}. We might also need to update the differential privacy budget over time \cite{chatzikokolakis2014predictive}. In this paper, we model one time sharing of a secret, assuming that the prior leakage and the differential privacy budget are properly computed and adjusted based on the previous observations. 

Our problem is also related to the problem of adversarial machine learning \cite{BarrenoNSJT06, HuangJNRT11} and the design of security mechanisms, such as intelligent spam detection algorithms \cite{LiuC09, BrucknerS11, KorzhykYKCT11}, against adaptive attackers. It is also similar to the problem of placing security patrols in an area to minimize the threat of attackers \cite{ParuchuriPMTOK08}, and faking location-based queries to protect against localization attack \cite{ShokriTTHL12}. The survey \cite{ManshaeiZABH11} explores more examples of the relation between security and game theory.

\section{Definitions}\label{sec:definitions}

In this section, we define different parts of our model. We assume a user shares her data through an information sharing system in order to obtain some service (utility). We also assume that users want to protect their sensitive information, while they share their data with untrusted entities. For example, in the case of sharing location-tagged data with a service provider, a user might want to hide the exact visited locations, their semantics, or her activities that can be inferred from the visited locations. We refer to the user's sensitive information as her {\em secret}. To protect her privacy, we assume that user obfuscates her data before sharing or publishing it. Figure~\ref{fig:framework} illustrates the information flow that we assume in this paper.

The input to the protection mechanism is a secret $s \in S$, where $S$ is the set of all possible values that $s$ can take (for example, the locations that the user can visit, or the individuals that she is acquainted with). Let prior leakage $\pi$ be the probability distribution over values of $s$ to reflect the data model and the a priori exposed information about the secret.
\begin{align}\label{eq:prior:pdf}
    \pi(s) = \pr{\rv{S} = s}
\end{align}

The probability distribution $\pi$ is estimated by the suer to be the predictability of the user's secret given her exposed information in the past. Thus, anytime that user shares some (obfuscated) information, she needs to update this probability distribution \cite{theodorakopoulos2014prolonging, danezis2013you}. This is how we incorporate the correlation between users' data shared over time.

\subsection{Obfuscation Mechanism}

We assume that a user wants to preserve her privacy with respect to $s$.  To protect her privacy, a user obfuscates her secret $s$ and shares an inaccurate version of it through the system.  We assume that this obfuscated data $o \in O$ is observable through the system. We consider a generic class of obfuscation mechanisms, in which the observable $o$ is sampled according to the following probability distribution.
\begin{align}\label{eq:obfuscation:pdf}
    p(o|s) = \pr{\rv{O} = o | \rv{S} = s}
\end{align}

Thus, we model the privacy preserving mechanism as a noisy channel between the user and the untrusted observer. This is similar to the model used in quantitative information flow and quantitative side-channel analysis \cite{KopfB07, AlvimCPS12}. The output, i.e., the set of observables $O$, can in general be a member of the powerset of $S$. As an example, in the most basic case, $O = S$, i.e., the protection mechanism can only perturb the secret by replacing it with another possible secret's value. This can happen through adding noise to $s$. In a more generic case, the members of $O$ can contain a subset of secrets. For example, the protection mechanism can generalize a location coordinate, by reducing its granularity.

\begin{figure}[t]
\begin{tikzpicture} [auto=left,scale=0.45]
    \node [rectangle,draw,thick,label=above:{\em\footnotesize prior leakage}] (p) at (1, 1) {$\pi(s)$};
    \node [circle,draw,thick,label=below:{\em\footnotesize secret}] (ar) at (4, 1) {$s$};
    \node [rectangle,draw,thick,blue,label=above:{\em\footnotesize obfuscation}] (f) at (7, 1) {$p(o | s)$};
    \node [circle,draw,thick,label=below:{\em\footnotesize observable}] (or) at (10, 1) {$o$};
    \node [rectangle,draw,thick,blue,label=above:{\em\footnotesize inference}] (h) at (13, 1) {$q(\hat{s} | o)$};
    \node [circle,draw,thick,label=below:{\em\footnotesize estimate}] (er) at (16, 1) {$\hat{s}$};
    \node [label=below:{\em\footnotesize utility cost}] (dq) at (7, -2) {$c(o, s)$};
    \node [label=above:{\em\footnotesize user-specific privacy}] (dp) at (13, 5)  {$d(\hat{s}, s)$};

    \path [draw,->] (p)  -- (ar);
    \path [draw,->] (ar) -- (f);
    \path [draw,->] (f)  -- (or);
    \path [draw,->] (or) -- (h);
    \path [draw,->] (h)  -- (er);

    \path [red] (dp) edge [<-,bend right] node {} (ar)
                     edge [<-,bend left]  node {} (er)
                (dq) edge [<-,bend left]  node {} (ar)
                     edge [<-,bend right] node {} (or);
\end{tikzpicture}
\caption{The Information Sharing Framework. Probability distribution $\pi$ encodes the user's estimation of a priori leaked information about secret $s$. The secret is obfuscated by the protection mechanism $p$ whose output is an observable $o$. The adaptive adversary (anticipated by the user) runs inference attack $q$ on $o$ and draws a probability distribution over estimates $\hat{s}$. Distance function $c$ denotes the utility cost of the protection mechanism due to obfuscation. Distance function $d$ denotes the privacy of user (for distortion privacy metric) or the required indistinguishability between secrets (for differential privacy metric). User defines the distance function $d$ to reflect her privacy sensitivities.} \label{fig:framework}
\end{figure}
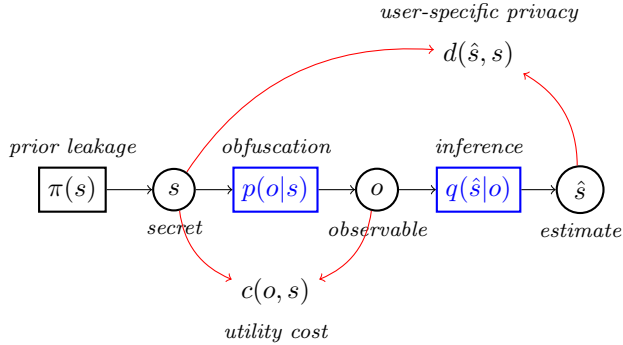

\subsection{Utility Cost}

Users incur a utility loss due to obfuscation. Let the distance function $c(o, s)$ determine the utility cost (information usefulness degradation) due to replacing a secret $s$ with an observable $o$. The cost function is dependent on the application of the shared information, on the specific service that is provided to the user, and also on the user's expectations. We compute the expected utility cost of a protection mechanism $p$ as
\begin{align}\label{eq:utility:avgmetric}
    \sum_s \pi(s) \sum_o p(o|s) \cdot c(o, s).
\end{align}

We can also compute the worst (maximum) utility cost over all possible secrets as
\begin{align}\label{eq:utility:maxmetric}
    \maximize{s} \sum_o p(o|s) \cdot c(o, s).
\end{align}

In this work, we do not plan to determine which metrics are the best representative utility loss metrics for different types of services or users. We only assume that the designer of optimal obfuscation mechanism is provided with such a utility function, for example, by constructing it according to the application \cite{micinski2013empirical}, or by learning it automatically \cite{bilogrevic2015predicting} from the users' preferences and application profile.

\subsection{Inference Attack}

We stated that the user wants to protect her privacy with respect to secret $s$ against untrusted observers. To be consistent with this, we define the adversary as an entity who aims at finding the user's secret by observing the outcome of the protection mechanism and minimizing the user's privacy with respect to her privacy sensitivities. For any observation $o$, then we determine the probability distribution over the possible secrets $\hat{s} \in S$ as to be the true secret of the user.
\begin{align}\label{eq:attack:pdf}
    q(\hat{s}|o) = \pr{\rv{S} = \hat{s} | \rv{O} = o}
\end{align}

The goal of the inference algorithm $q$ is to invert a given protection mechanism $p$ to estimate $\hat{s}$. The error of adversary, in this estimation process, determines the effectiveness of the inference algorithm, which is captured by the distortion privacy metric. 

\subsection{Distortion Privacy Metric}

As stated above, the user's privacy and the adversary's inference error are two sides of the same coin. We define the privacy gain of the user with secret $s$ as a distance between the two data points: $d(\hat{s}, {s})$, where $\hat{s}$ is the a posteriori estimation of the secret \cite{ShokriTLH11}. The distance function $d$ is determined by the sensitivity of the user towards each secret $s$ when estimated as $\hat{s}$. A user would be less worried about revealing $o \sim p(o|s)$, if the portrait of her secret $s$ in the eyes of adversary is an estimate $\hat{s}$ with a large distance $d(\hat{s}, s)$.

This distance function is defined by the user. It could be a semantic distance between different values of secrets to reflect the privacy risk of $\hat{s}$ on user when her secret is $s$. Usually, the highest risk is associated with the case where the estimate $\hat{s}$ is equal to the secret $s$. However, sometimes even wrong estimates can impose a high risk on the user, for example when they leak information about the semantic of the secret.

We compute the user privacy obtained through a protection mechanism $p$, with respect to a given inference algorithm $q$, for a specific secret $s$ as
\begin{align}\label{eq:privacy:bayesian:persecret}
    \sum_o p(o|s) \sum_{\hat{s}} q(\hat{s}|o) \cdot d(\hat{s}, s).
\end{align}

By averaging this value over all possible secrets, we compute the expected distortion privacy of the user as
\begin{align}\label{eq:privacy:bayesian:average}
    \sum_s \pi(s) \sum_o p(o|s) \sum_{\hat{s}} q(\hat{s}|o) \cdot d(\hat{s}, s).
\end{align}

This metric shows the average estimation error, or how distorted the reconstructed user's secret is. Thus, we refer to it as the {\em distortion} privacy metric. 

What associates a semantic meaning to this metric is the distance function $d$. Many distance functions can be defined to reflect distortion privacy. This depends on the type of the secret and to the sensitivity of the user. For example, if the user's secret is her age, function $d$ could be the absolute distance between two numbers. If the secret is the user's location, function $d$ could be a Euclidean distance between locations, or their semantic dissimilarity. If the secret is the movies that she has watched, function $d$ could be the Jaccard distance between two sets of movies. 

\subsection{Differential Privacy Metric}

The privacy that is achieved by an obfuscation mechanism can be computed with respect to the information leakage through the mechanism, regardless of the secret's inference. For example, the differential privacy metric, originally proposed for protecting privacy in statistical databases \cite{Dwork06}, is sensitive only to the difference between the probabilities of obfuscating multiple secrets to the same observation (which is input to the attack).

According to the original definition of differential privacy, a randomized function $\mathcal{K}$ (that acts as the privacy protection mechanism) provides $\epsilon$-differential privacy if for all data sets $D$ and $D'$, that differ on at most one element, and all $Y \subseteq Range(\mathcal{K})$, the following inequality holds.
\begin{align}\label{eq:privacy:diff:originalnotation}
    \pr{\mathcal{K}(D) \in Y} \leq \mathrm{exp}(\epsilon) \cdot \pr{\mathcal{K}(D') \in Y}
\end{align}

Differential privacy is not limited to statistical databases. It has been used in many different contexts where various types of adjacency relations capture the context dependent privacy. A typical example is edge privacy in graphs \cite{nissim2007smooth}. It has also been proposed for arbitrary distance function between secrets \cite{ChatzikokolakisABP13}.

This notion can simply be used for measuring information leakage \cite{AlvimACP11}. It has been shown that differential privacy imposes a bound on information leakage \cite{alvim2011relation, alvim2012differential}. And, this is exactly why we are interested in this metric. Let $d^\epsilon(s, s')$ be a distinguishability metric between $s, s' \in S$. A protection mechanism is defined to be differentially private if for all secrets $s, s' \in S$, where $d^\epsilon(s, s') \leq d^\epsilon_m$, and all observables $o \in O$, the following inequality holds.
\begin{align}\label{eq:privacy:diff:original}
    p(o | s) \leq \mathrm{exp}(\epsilon) \cdot p(o | s')
\end{align}

In this paper, we use a generic definition of differential privacy, assuming arbitrary distance function $d^\epsilon()$ on the secrets \cite{BartheKOZ12, ChatzikokolakisABP13, DworkMNS06, ReedP10}. In this form, a protection mechanism is differentially private if for all secrets $s, s' \in S$, with distinguishability $d^\epsilon(s, s')$, and for all observables $o \in O$, the following holds.
\begin{align}\label{eq:privacy:diff:gneric:mult}
    p(o | s) \leq \mathrm{exp}(\epsilon \cdot d^\epsilon(s, s')) \cdot p(o | s')
\end{align}

In fact, the differential privacy metric guarantees that, given the observation, there is not enough convincing evidence to prefer one secret to other similar ones (given $d^\epsilon$). In other words, it makes multiple secret values indistinguishable from each other. 

\section{Problem Statement}\label{sec:problem}

The problem that we address in this paper is to find an optimal balance between privacy and utility, and to construct the protection mechanisms that achieve such optimal points. More precisely, we want to construct utility-maximizing obfuscation mechanisms with joint differential-distortion privacy guarantees. 

The problem is to find a probability distribution function $p^*$ such that it minimizes utility cost of the user, on average,
\begin{align}\label{eq:objective:utility:avgmetric}
    p^* = \argmin{p} \sum_s \pi(s) \sum_o p(o|s) \cdot c(o, s)
\end{align}
or, alternatively, over all the secrets
\begin{align}\label{eq:objective:utility:maxmetric}
    p^* = \argmin{p} \,\,\, \maximize{s} \sum_o p(o|s) \cdot c(o, s)
\end{align}
under the user's privacy constraints.

\subsection{Distortion Privacy Constraint}

Let $d_m$ be the minimum desired distortion privacy level. The user's average distortion privacy is guaranteed if the obfuscation mechanism $p^*$ satisfies the following inequality.
\begin{align}\label{eq:constraint:privacy:bayesian}
    \sum_s \pi(s) \sum_o p^*(o|s) \sum_{\hat{s}} q^*(\hat{s}|o) \cdot d(\hat{s}, s) \ge d_m
\end{align}
where $q^*$ is the optimal inference attack against $p^*$.

\subsection{Differential Privacy Constraint}

Let $\epsilon_m$ be the differential privacy budget associated with the minimum desired privacy of the user, and $d^\epsilon_m$ be the distinguishability threshold. The user's privacy is guaranteed if $p^*$ satisfies the following inequality.
\begin{align}\label{eq:constraint:privacy:diff:mult2}
    p^*(o | s) \leq \mathrm{exp}(\epsilon_m) \cdot p^*(o | s'), \, \forall o, s, s': d^\epsilon(s, s') \leq d^\epsilon_m
\end{align}

Or, alternatively (following \cite{ChatzikokolakisABP13}'s definition of differential privacy):
\begin{align}\label{eq:constraint:privacy:diff:mult}
    p^*(o | s) \leq \mathrm{exp}(\epsilon_m \cdot d^\epsilon(s, s')) \cdot p^*(o | s'), \, \forall s, s', o
\end{align}

In this paper, we mainly use the latter definition, but make use of the former one as the basis to reduce the computation cost of optimizing differential privacy (see Appendix~\ref{sec:approx}).

\section{Solution: Privacy Games}\label{sec:game}

The flow of information starts from the user where the secret is generated. The user then selects a protection mechanism, and obfuscates her secret according to its probabilistic function. After the adversary observes the output, he can design an optimal inference attack against the obfuscation mechanism to invert it and estimate the secret. We assume the obfuscation mechanism is not oblivious and is known to the adversary. This gives the adversary the upper hand against the user in their conflict. So, designing an obfuscation mechanism against a fixed attack is always suboptimal.

The best obfuscation mechanism is the one that {\em anticipates} the adversary's attack. Thus, the obfuscation mechanism should be primarily designed against an {\em adaptive} attack which is tailored to each specific obfuscation mechanism. So, by assuming that the adversary designs the best inference attack against each protection mechanism, the user's goal (as the defender) must be to design the obfuscation mechanism that maximizes her (privacy or utility) objective against an adversary that optimizes the conflicting objective of guessing the user's secret. The adversary is an entity assumed by the user as the entity whose objective s exactly the opposite of the user's. So, we do not model any particular attacker but the one that minimizes user's privacy according to distance functions $d$ and $d^{\epsilon}$.

For each obfuscation mechanism there is an inference attack that optimizes the adversary's objective and leads to a certain privacy and utility payoff for the user. The optimal obfuscation mechanism for the user is the one that brings the maximum payoff for her, against the mechanism's corresponding optimal inference attack. 

Enumerating all pairs of user-attacker mechanisms to find the optimal obfuscation function is infeasible. We model the joint user-adversary optimization problem as a leader-follower (Stackelberg) game between the user and the adversary. The user leads the game by choosing the protection mechanism $p$, and the adversary follows by designing the inference attack $q$. The solution to this game is the pair of user-adversary best response strategies $p^*$ and $q^*$ which are mutually optimal against each other. If the user implements $p^*$, we have already considered the strongest attack $q^*$ against it. Thus, $p^*$ is robust against {\em any} algorithm used as inference attack.

For any secret $s \in S$, the strategy space of the user is the set of observables $O$. For any observable $o \in O$, the strategy space of the adversary is the set of secrets $S$ (all possible adversary's estimates $\hat{s} \in S$). For a given secret $s \in S$, we represent a mixed strategy for the user by a vector $p(.|s) = \left( p(o_1|s), p(o_2|s), \cdots, p(o_m|s) \right)$, where $\{o_1, o_2, \cdots, o_m\} = O$. Similarly, a mixed strategy for the adversary, for a given observable $o \in O$ is a vector $q(.|o) = \left( q(\hat{s}_1|o), q(\hat{s}_2|o), \cdots, q(\hat{s}_n|o) \right)$, where $\{\hat{s}_1, \hat{s}_2, \cdots, \hat{s}_n\} = S$. Note that the vectors $p(.|s)$ and $q(.|o)$ are respectively the conditional distribution functions associated with an obfuscated function for a secret $s$ and an inference algorithm for an observable $o$. Let $P$ and $Q$ be the sets of all mixed strategies of the user and the adversary, respectively.
\begin{align}
    P = \{ & p(.|s) = ( p(o_1|s), p(o_2|s), \cdots, p(o_m|s) ), \forall s \in S \,: \nonumber\\
    & p(o_i|s) \ge 0, \forall o_i \in O, \sum_i p(o_i|s) = 1
    \} \label{eq:game:mixedstrategies:p} \\
    Q = \{ & q(.|o) = ( q(\hat{s}_1|o), q(\hat{s}_2|o), \cdots, q(\hat{s}_n|o) ), \forall o \in O \,: \nonumber\\
    & q(\hat{s}_j|o) \ge 0, \forall \hat{s}_j \in S, \sum_j q(\hat{s}_j|o) = 1
    \} \label{eq:game:mixedstrategies:q}
\end{align}

A member vector of sets $P$ or $Q$ with a $1$ for the $k$th component and zeros elsewhere is the pure strategy of choosing action $k$. For example, an obfuscation function $p(.|s)$ for which $p(o_i|s) = 0, \forall i \neq k$ and $p(o_k|s) = 1$ is the pure strategy of exclusively and deterministically outputting observable $o_k$ for secret $s$. Thus, the set of pure strategies of a player is a subset of mixed strategies of the player.

In the case of the distortion privacy metric, the game needs to be formulated as a {\em Bayesian Stackelberg game}. In this game, we assume the probability distribution $\pi$ on the secrets and we find $p^* \in P$ and $q^* \in Q$ that create the equilibrium point. If user deviates from this strategy and chooses $p' \ne p^*$, there would be an inference attack $q'^*$ against it such that $(p', q'^*)$ leads to a lower privacy for the user, i.e., $p^*$ is optimal.

In the case of a differential privacy metric, as the metric is not dependent to the adversary's inference attack, the dependency loop between finding optimal $p^*$ and $q^*$ is broken. Nevertheless, it is still the user who plays first by choosing the optimal protection mechanism. In the following sections, we solve these games and provide solutions on how to design the optimal user-adversary strategies.

\section{Stackelberg Privacy Games}\label{sec:solution:bayesian}

Assume that the nature draws secret $s$ according to the probability distribution $\pi(s)$. Given $s$, the user draws $o$ according to her obfuscation mechanism $p(o|s)$, and makes it observable to the adversary. Given observation $o$, the adversary draws $\hat{s}$ according to his inference attack $q(\hat{s}|o)$. We assume that $\pi(s)$ is known to both players.  We want to find the mutually optimal $\langle p^*, q^* \rangle$: The solution of the Bayesian Stackelberg privacy game.

To this end, we first design the optimal inference attack against any given protection mechanism $p$.  This will be the {\em best response} of the adversary to the user's strategy. Then, we design the optimal protection mechanism for the user according to her objective and constraints, as stated in Section~\ref{sec:problem}.  This will be the user's best utility-maximizing strategy that anticipates the adversary's best response. 

\subsection{Optimal Inference Attack}

The adversary's objective is to minimize (the user's privacy and thus) the inference error in estimating the user's secret. Given a secret $s$, the distance function $d(\hat{s}, s)$ determines the error of an adversary in estimating the secret as $\hat{s}$. In fact, this distance is exactly what a user wants to maximize (or put a lower bound on) according to the distortion privacy metric. We compute the expected error of the adversary as
\begin{align}\label{eq:expectederror}
    & \sum_s \pi(s) \sum_{\hat{s}} \pr{\hat{s}|s} \cdot d(\hat{s}, s) = \nonumber\\
    & = \sum_{s, o, \hat{s}} \pi(s) \cdot p(o|s) \cdot q(\hat{s}|o) \cdot d(\hat{s}, s)
\end{align}

Therefore, we design the following linear program, through which we can compute the adversary's inference strategy that, given the probability distribution $\pi$ and obfuscation $p$, minimizes his expected error with respect to a distance function $d$.
\begin{subequations}\label{eq:lp:adversary:bayesian}
\begin{align}
    q^* =
    \argmin{q} & \sum_{s, o, \hat{s}} \pi(s) \cdot p(o|s) \cdot q(\hat{s}|o) \cdot d(\hat{s}, s) \label{eq:lp:adversary:bayesian:obj} 
\end{align}
\end{subequations}
under the constraint that the solution is a proper conditional probability distribution function. 

In the next subsection, we will show that the optimal deterministic inference (that associates one single estimate with probability one to each observation) results in the same privacy for the user \eqref{eq:proof:q}. Alternative ways to formulate this problem is given in Appendix~\ref{sec:optimalattack}.

\subsection{Optimal Protection Mechanism}

In this case, we assume the user would like to minimize her utility cost \eqref{eq:objective:utility:avgmetric} under a (lower bound) constraint on her privacy \eqref{eq:constraint:privacy:bayesian}. Therefore, we can formulate the problem as 
\begin{subequations}\label{eq:lp:user:utility-privacy:nested}
\begin{align}
    p^* =
    \argmin{p} & \sum_{s, o} \pi(s) \cdot p(o|s) \cdot c(o, s) \\
    \st & \sum_{s, o, \hat{s}} \pi(s) \cdot p(o|s) \cdot q^*(\hat{s}|o) \cdot d(\hat{s}, s) \geq d_m \label{eq:lp:user:utility-privacy:nested:c1}
\end{align}
\end{subequations}

However, solving this optimization problem requires us to know the optimal $q^*$ against $p^*$, for which we need to know $p^*$ as formulated in \eqref{eq:lp:adversary:bayesian}. So, we have two linear programs (one for the user and one for the adversary) to solve. But, the solution of each one is required in solving the other. This optimization dependency loop reflects the game-theoretic concept of {\em mutual best response} of the two players. This game is a {\em nonzero-sum Stackelberg game} as the user (leader player) and adversary (follower player) have different optimization objectives (one maximizes utility, and the other minimizes privacy). We break the dependency loop between the optimization problems using the game-theoretic modeling, and we prove that the user's best strategy can be constructed using linear programming.
\begin{theorem}
Given a probability distribution $\pi$, the distance functions $d$ and $c$, and the threshold $d_m$, the solution to the following linear program is the optimal protection strategy $p^*$ for the user, which is the solution to \eqref{eq:lp:user:utility-privacy:nested} with respect to adversary's best response \eqref{eq:lp:adversary:bayesian}.
\begin{subequations}\label{eq:lp:user:utility-privacy:game}
\begin{align}
    p^* =
    \argmin{p} & \sum_{s, o} \pi(s) \cdot p(o|s) \cdot c(o, s) \\
    \st
    & \sum_s \pi(s) \cdot p(o|s) \cdot d(\hat{s}, s) \ge x(o), \forall o, \hat{s} \label{eq:lp:user:utility-privacy:game:c1}\\
    & \sum_o x(o) \ge d_m \label{eq:lp:user:utility-privacy:game:c2} 
\end{align}
\end{subequations}
\end{theorem}

\begin{proof}
See Appendix~\ref{sec:proof}.
\end{proof}

\section{Optimal Differential Privacy}\label{sec:solution:differential}

In this section, we design optimal differentially private protection mechanisms. We solve the optimization problems for maximizing utility under privacy constraint. 

We design the following linear program to find the user strategy $p^*$ that guarantees user differential privacy \eqref{eq:constraint:privacy:diff:mult}, for a maximum privacy budget $\epsilon_m$, and minimizes the utility cost \eqref{eq:objective:utility:avgmetric} of the obfuscation mechanism.
\begin{subequations}\label{eq:lp:user:utility-privacy:diff:mult}
\begin{align}
    \minimize{p} & \sum_{s,o} \pi(s) \cdot p(o|s) \cdot c(o,s) \label{eq:lp:user:utility-privacy:diff:mult:obj} \\
    \st & \frac{p(o|s)}{p(o|s')} \leq \cdot \mathrm{exp}(\epsilon_m \cdot d^\epsilon(s, s')) \ , \forall s, s', o 
\end{align}
\end{subequations}

Or, alternatively, for a distinguishability bound $d^\epsilon_m$, we can solve the following.
\begin{subequations}\label{eq:lp:user:utility-privacy:diff:mult2}
\begin{align}
    \minimize{p} & \sum_{s,o} \pi(s) \cdot p(o|s) \cdot c(o,s) \label{eq:lp:user:utility-privacy:diff:mult:obj} \\
    \st & \frac{p(o|s)}{p(o|s')} \leq \mathrm{exp}(\epsilon_m) \ , \forall o, s, s': d^\epsilon(s, s') \leq d^\epsilon_m 
\end{align}
\end{subequations}

\section{Optimal Joint Differential and Distortion Privacy Mechanism} \label{sec:solution:joint}

Obfuscation mechanisms designed based on distortion and differential privacy protect the user's privacy from two different angles. In general, for arbitrary $d$ and $d^{\epsilon}$, there is no guarantee that a mechanism with a bound on one metric holds a bound on the other.

Distortion privacy metric reflects the {\em absolute} privacy of the user, based on the posterior estimation on the obfuscated information. Differential privacy metric reflects the {\em relative} information leakage of each observation about the secret. However, it is not a measure on the extent to which the observer, who already has some knowledge about the secret from the previously shared data, can guess the secret correctly. So, the inference might be very accurate (because of the background knowledge) despite the fact that the obfuscation in place is a differentially-private mechanism.

As distortion and differential metrics guarantee different dimensions of the user's privacy requirements, we respect both in a protection mechanism. This assures that not only the information leakage is limited, but also the absolute privacy level is at the minimum required level. Thanks to our unified formulation of privacy optimization problems as linear programs, the problem of jointly optimizing and guaranteeing privacy with both metrics can also be formulated as a linear program.

The solution to the following linear program is a protection mechanism $p^*$ that maximizes the user's utility and guarantees a minimum distortion privacy $d_m$ and a minimum differential privacy $\epsilon_m$, given probability distribution $\pi$ and distance functions $c$ and $d$ and distinguishability metric $d^\epsilon$. The value of the optimal solution is the utility cost of the optimal mechanism.
\begin{subequations}\label{eq:lp:user:utility-jointprivacy}
\begin{align}
    \minimize{p} & \sum_{s, o} \pi(s) \cdot p(o|s) \cdot c(o, s) \\
    \st
    & \sum_s \pi(s) \cdot p(o|s) \cdot d(\hat{s}, s) \ge x(o), \forall o, \hat{s} \\
    & \sum_o x(o) \ge d_m \\
    & \frac{p(o|s)}{p(o|s')} \leq \cdot \mathrm{exp}(\epsilon_m \cdot d^\epsilon(s, s')) \ , \forall s, s', o 
\end{align}
\end{subequations}

\section{Analysis}\label{sec:analysis}

We have implemented all our linear program solutions in a software tool that can be used to process data for different applications, in different settings. In this section, we use our tool to design privacy protection mechanisms, and also to make a comparison between different optimal mechanisms, i.e., distortion, differential, and joint distortion-differential privacy preserving mechanisms. We study the properties of these mechanisms and we show how robust they are with respect to inference attack algorithms as well as to the adversary's knowledge on secrets. We also investigate their utility cost for protecting privacy. Furthermore, we show that the optimal joint distortion-differential mechanisms are more robust than the two mechanisms separately.  In Appendix~\ref{sec:approx}, we discuss and evaluate approximations of the optimal solution for large number of constraints.

We run experiments on location data, as today they are included in most of data sharing applications. We use a real data-set of location traces collected through the Nokia Lausanne Data Collection Campaign \cite{kiukkonen2010towards}. The location information belong to a $15 \times 8$km area. We split the area into $20 \times 15$ cells. We consider location of a mobile user in a cell as her secret. Hence, the set of secrets is equivalent to the set of location cells. We assume the set of observables to be the set of cells, so the users obfuscate their location by perturbation (i.e., replacing their true location with any location in the map). We run our experiments on $10$ randomly selected users, to see the difference in the results due to difference in user's location distribution $\pi$ based on users' different location access profiles. We build $\pi$ for each user separately given their individual location traces, using maximum likelihood estimation (normalizing the user's number of visits to each cell in the tarce).

We assume a Euclidean distance function for $d$ and $d^\epsilon$. This reflects the sensitivity of user towards her location. By using this distance function for distortion privacy, we guarantee that the adversary cannot guess the user's true location with error lower than the required privacy threshold ($d_m$). Choosing Euclidean distance function as the metric for distinguishability ensures that the indistinguishability between locations is larger for locations that are located closer to each other. 

We assume a Hamming distortion function for $c$ (i.e., the utility cost is $0$ only if the user's location and the observed location are the same, otherwise the cost is $1$). The utility metric can vary depending on the location-based sharing application and also the purpose for which the user shares her location \cite{bilogrevic2015predicting}. Choosing the Hamming function reflects the utility requirement of users who want to inform others about their current location in location check-in applications.

\begin{figure*}[t!]
   \centering
   \subfigure[Achieved distortion privacy for an optimal differential privacy mechanism with $\epsilon_m$. Each line corresponds to one user.]{
      \includegraphics[width=0.98\columnwidth]{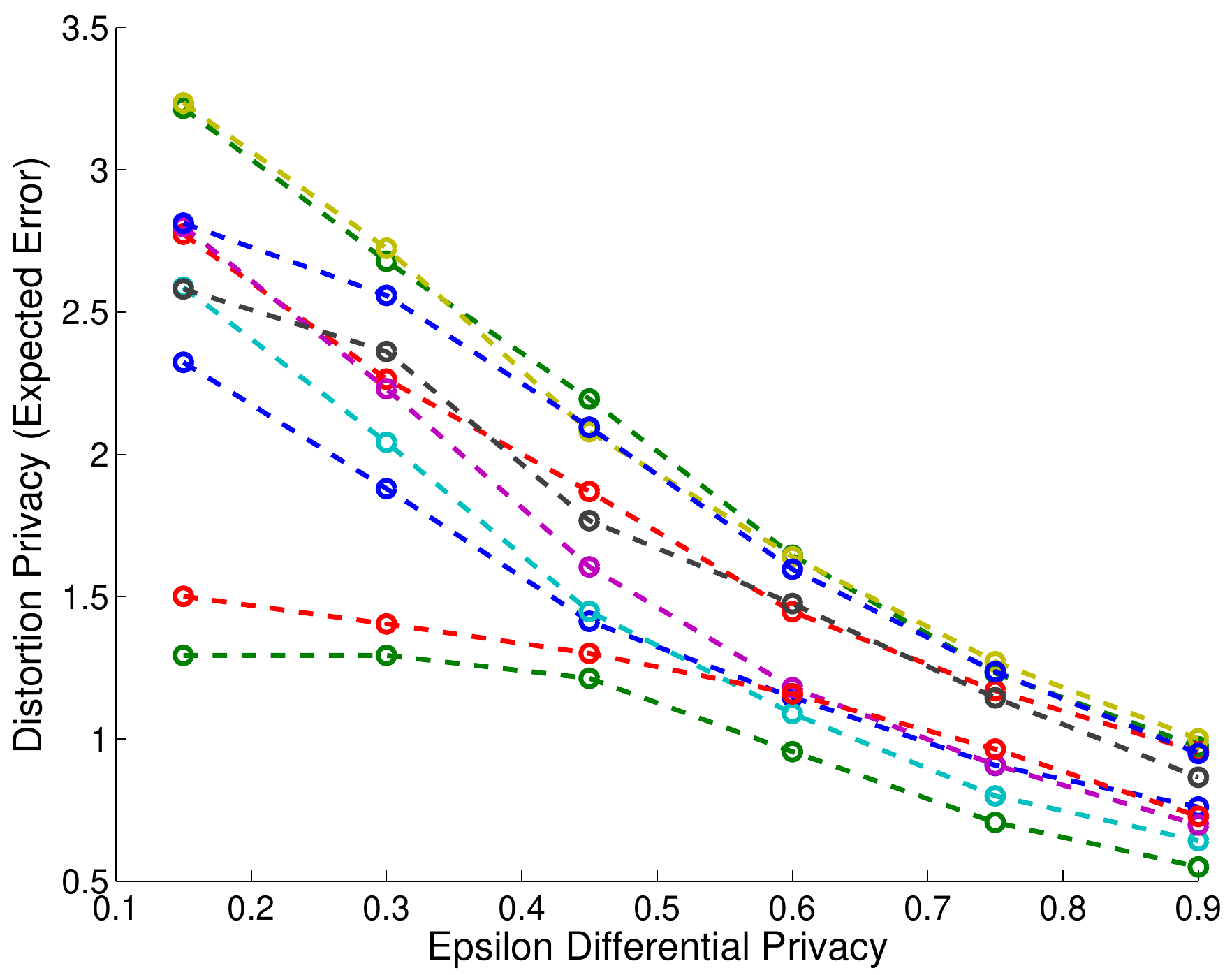}
      \label{fig:privacy_loop_epsilon_vs_privacy_optimal}
   }~~
   \subfigure[Utility cost metric versus distortion privacy metric, for three different optimal obfuscation mechanisms]{
      \includegraphics[width=0.98\columnwidth]{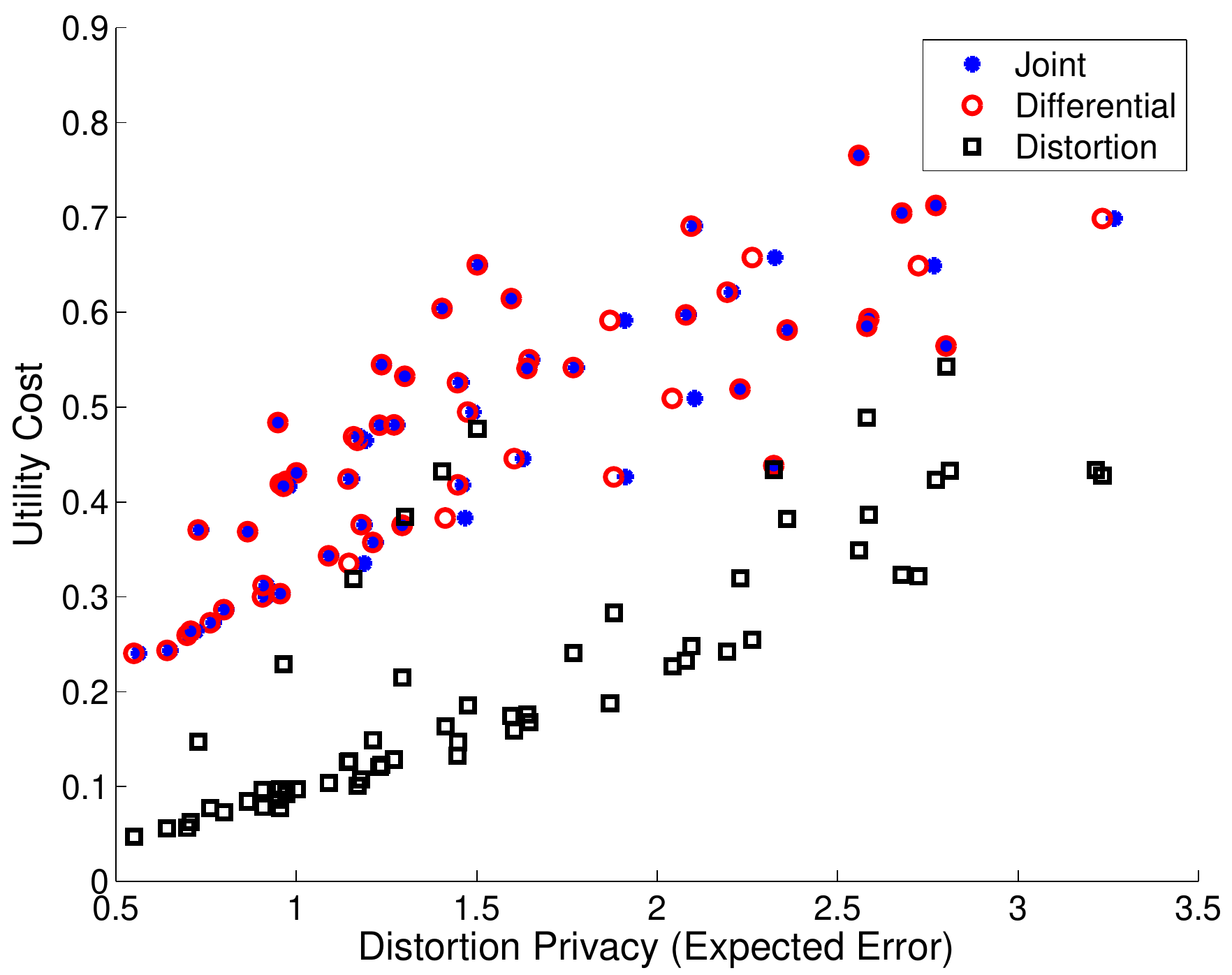}
      \label{fig:privacy_loop_privacy_vs_cost_optimalattack}
   }
   \centering
\caption{Privacy and Utility of optimal protection mechanisms, computed based on the evaluation scenario number 1 in Section~\ref{sec:analysis:firsteval}. Each dot in the plots corresponds to privacy of one user for one value of $\epsilon_m$.
}
\label{fig:privacy_loop_optimalattack}
\end{figure*}

We evaluate utility-maximizing optimal protection mechanisms with three different privacy constraints:
\begin{itemize}
  \item {\em Distortion Privacy Protection}, \eqref{eq:lp:user:utility-privacy:game}.
  \item {\em Differential Privacy Protection}, \eqref{eq:lp:user:utility-privacy:diff:mult}.
  \item {\em Joint Distortion-Differential Privacy Protection}, \eqref{eq:lp:user:utility-jointprivacy}.
\end{itemize}

We compare the effectiveness of these protection mechanisms against inference attacks by using the distortion privacy metric \eqref{eq:privacy:bayesian:average}. We consider two inference attacks:
\begin{itemize}
  \item {\em Optimal Attack}, \eqref{eq:lp:adversary:bayesian}.
  \item {\em Bayesian Inference Attack}, using the Bayes rule:
      \begin{align}\label{eq:attack:bayesian}
    q(\hat{s}|o) = \frac{\pi(\hat{s}) \cdot p(o | \hat{s})}{\pr{o}} = \frac{\pi(\hat{s}) \cdot p(o | \hat{s})}{\sum_s \pi(s) \cdot p(o | s)}
\end{align}
\end{itemize}

\subsection{Comparing Obfuscation Mechanisms}\label{sec:analysis:firsteval}

\paragraph*{\em Scenario 1.} Our first goal is to have a fair comparison between optimal distortion privacy mechanism and optimal differential mechanism. To this end, we set the privacy parameter $\epsilon_m$ to $\{0.15, 0.3, \cdots, 0.9\}$. 

For each user and each value of $\epsilon_m$,
\begin{enumerate}
  \item We compute the optimal differential privacy mechanism using \eqref{eq:lp:user:utility-privacy:diff:mult}. Let $p^*_{\epsilon_m}$ be the optimal mechanism.
  \item We run optimal attack \eqref{eq:lp:adversary:bayesian} on $p^*_{\epsilon_m}$, and compute the user's absolute distortion privacy as $AP(p^*_{\epsilon_m})$.
  \item We compute the optimal distortion privacy mechanism $p^*_{d_m}$ using \eqref{eq:lp:user:utility-privacy:game}. For this, we set the privacy lower-bound $d_m$ to $AP(p^*_{\epsilon_m})$. This enforces the distortion privacy mechanism to guarantee what the differential privacy mechanism provides.
  \item We compute the optimal joint distortion-differential privacy mechanism $p^*_{\epsilon_m, d_m}$ using \eqref{eq:lp:user:utility-jointprivacy}. We set the privacy lower-bounds to $\epsilon_m$ and $d_m$ for the differential and distortion constraints, respectively.
  \item We run optimal attack \eqref{eq:lp:adversary:bayesian} on both $p^*_{d_m}$ and $p^*_{\epsilon_m, d_m}$, and compute the user's absolute distortion privacy as $AP(p^*_{d_m})$ and $AP(p^*_{\epsilon_m, d_m})$, respectively.
  \item As a baseline for comparison, we run Bayesian inference attack \eqref{eq:attack:bayesian} on the three optimal mechanisms $p^*_{\epsilon_m}$, $p^*_{d_m}$, and $p^*_{\epsilon_m, d_m}$. 
\end{enumerate}

Figure~\ref{fig:privacy_loop_optimalattack} shows the results of our analysis, explained above. Distortion privacy is measured in km and is equivalent to the expected error of adversary in correctly estimating location of users. Figure~\ref{fig:privacy_loop_epsilon_vs_privacy_optimal} shows how expected privacy of users $AP(p^*_{\epsilon_m})$ decreases as we increase the value of the lower-bound on differential privacy $\epsilon_m$. Users have different secret probability distribution, with different randomness. However, as $\epsilon_m$ increases, expected error of adversary (the location privacy of users) converges down to below $1$km. Figure~\ref{fig:privacy_loop_privacy_vs_cost_optimalattack} plots the utility cost versus distortion privacy of each optimal protection mechanism. As we have set the privacy bound of the optimal distortion mechanism (and of course the optimal joint mechanism) to the privacy achieved by the optimal differential mechanism, we can make a fair comparison between their utility costs. We observe that the utility cost for achieving some level of distortion privacy is much higher for optimal differential and joint mechanisms compared with the optimal distortion mechanism. Note that the utility cost of differential and joint mechanisms are the same. So, distortion privacy bound does not impose more cost than what is already imposed by the differential privacy mechanism. 

As we set $d_m$ to $AP(p^*_{\epsilon_m})$, the user's distortion privacy in using optimal distortion and optimal differential mechanism is the same, when we confront them with the optimal attack \eqref{eq:lp:adversary:bayesian}. In Figure~\ref{fig:privacy_loop_diff_vs_bayes_inferenceattack}, however, we compare the effectiveness of these two mechanisms against Bayesian inference attack \eqref{eq:attack:bayesian}. It is interesting to observe that the optimal differential mechanism is more robust to such attacks compared to the optimal distortion mechanisms. This explains the extra utility cost due to optimal differential mechanisms.

\begin{figure}[t]
   \centering
      \includegraphics[width=0.98\columnwidth]{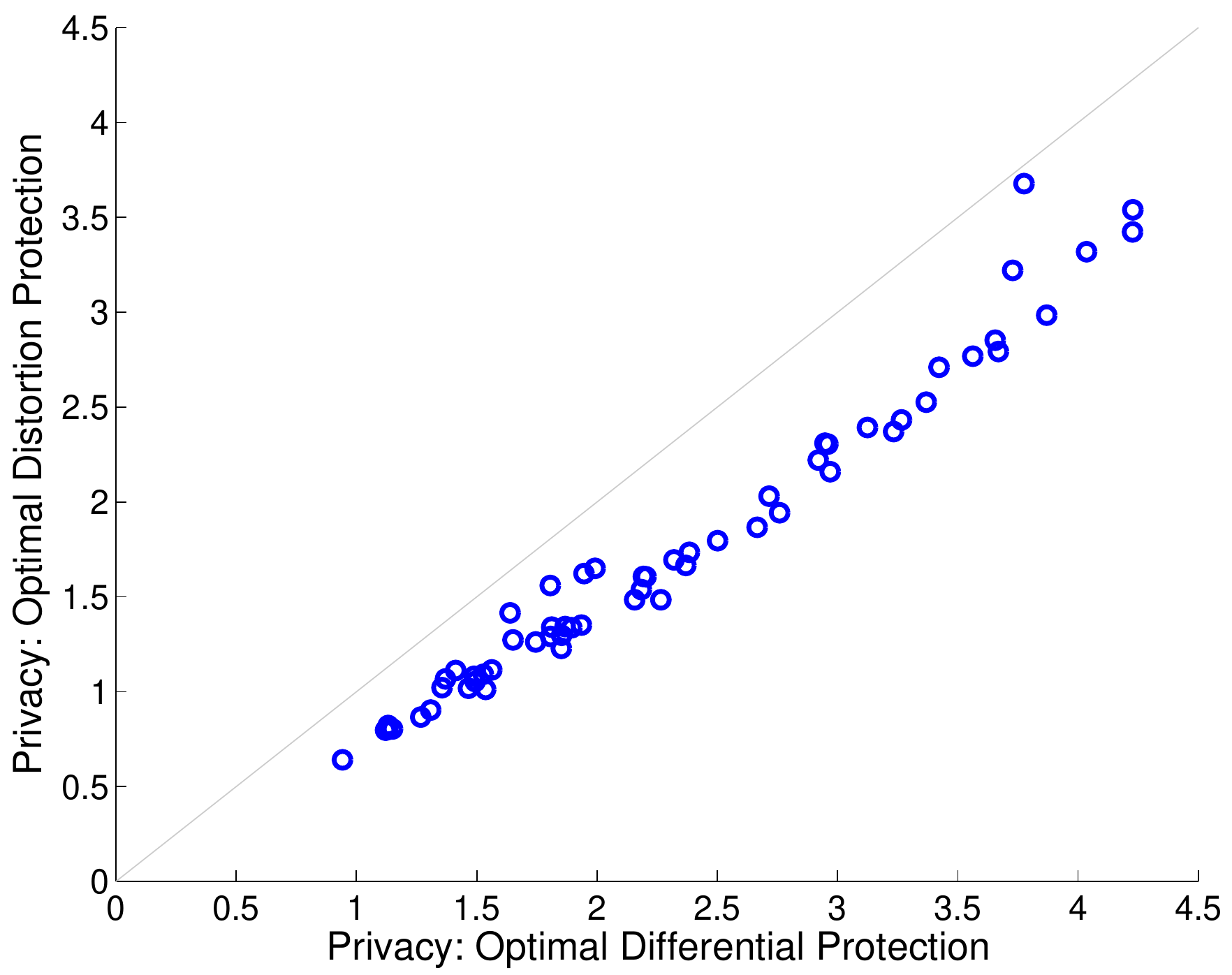}
      \caption{Distortion privacy of users against the Bayesian inference attack \eqref{eq:attack:bayesian} when using optimal differential privacy obfuscation versus using optimal distortion privacy obfuscation. Each dot represents privacy of one user for one value of $\epsilon_m$.}
      \label{fig:privacy_loop_diff_vs_bayes_inferenceattack}
\end{figure}
\begin{figure}[t]
   \centering
      \includegraphics[width=0.98\columnwidth]{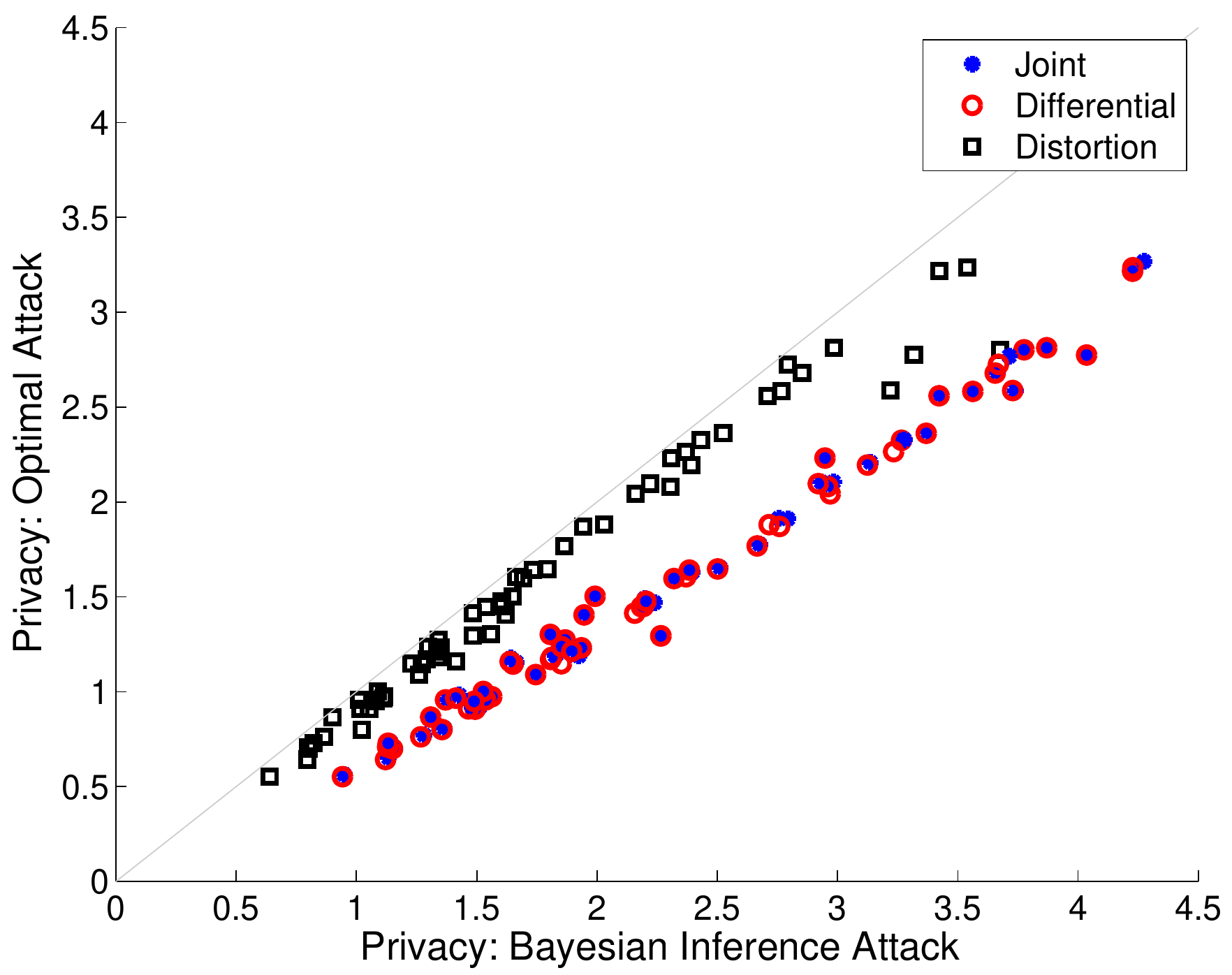}
      \caption{Distortion privacy of users using any of the three optimal mechanisms against the basic Bayesian inference attack \eqref{eq:attack:bayesian} versus their privacy against the optimal attack \eqref{eq:lp:adversary:bayesian}. Each dot represents privacy of one user for one value of $\epsilon_m$.}
      \label{fig:privacy_loop_inference_vs_optimal}
\end{figure}
\begin{figure}[t]
   \centering
      \includegraphics[width=0.98\columnwidth]{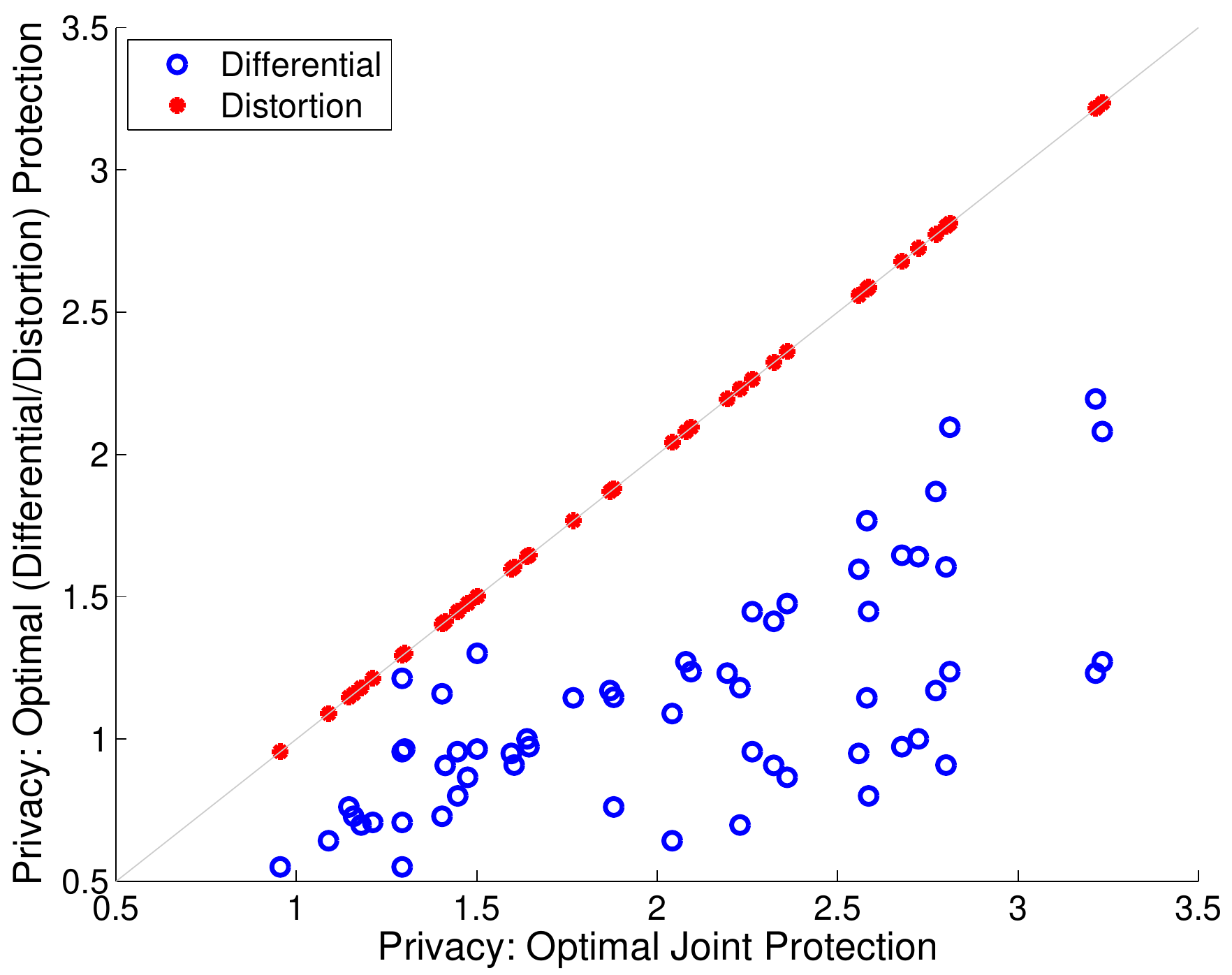}
      \caption{Users' privacy using optimal differential or distortion protection mechanism versus using optimal joint protection mechanism. Distortion privacy is computed using optimal attack \eqref{eq:lp:adversary:bayesian}.
      }
      \label{fig:privacy_alljoint_joint_vs_diffbayes_optimalattack}
\end{figure}

In Figure~\ref{fig:privacy_loop_inference_vs_optimal}, we compare the effectiveness of Bayesian inference attack \eqref{eq:attack:bayesian} and optimal attack \eqref{eq:lp:adversary:bayesian}. We show the results for all three optimal protection mechanisms. It is clear that optimal attack outperforms the Bayesian attack, as users have a relatively higher privacy level under the Bayesian inference. However, the difference is more obvious for the case of differential protection and joint protection mechanisms. The Bayesian attack overestimates users' privacy, as it ignores the distance function $d$, whereas the optimal attack minimizes the expected value of $d$ over all secrets and estimates.

\paragraph*{\em Scenario 2.} In this paper, we introduce the optimal joint distortion-differential protection mechanisms to provide us with the benefits of both mechanisms. Figure~\ref{fig:privacy_loop_privacy_vs_cost_optimalattack} shows that the optimal joint mechanism is not more costly than the two optimal distortion and differential mechanisms. It also shows that it guarantees the highest privacy for a certain utility cost. To further study the effectiveness of optimal joint mechanisms, we run the following evaluation scenario. 

We design optimal differential mechanisms for some values of $\epsilon_m$. And, we design optimal distortion mechanisms for some values of $d_m$ that are higher than the distortion privacy resulted from those differential privacy mechanisms. We also construct their joint mechanisms given the $\epsilon_m$ and $d_m$ parameters. Figure~\ref{fig:privacy_alljoint_joint_vs_diffbayes_optimalattack} shows how the optimal joint mechanism adapts itself to guarantee the maximum of the privacy levels guaranteed by optimal Bayesian and optimal differential mechanisms individually. This is clear from the fact that users' privacy for the optimal joint mechanism is equal to their privacy for distortion mechanism (that as we set in our scenario, they are higher than that of differential mechanisms). 

Thus, by adding the distortion privacy constraints in the design of optimal mechanisms, we can further increase the privacy of users (with the same utility cost) that cannot be otherwise achieved by only using differential mechanisms.

\paragraph*{\em Scenario 3.} In order to further investigate the relation between the privacy (and utility) outcome of the optimal joint mechanism and that of individual differential or distortion privacy mechanisms, we run the following set of experiments on all the available user profiles. 

\begin{enumerate}
  \item For any value of $\epsilon_m$ in $\{0.2, 0.4, \cdots, 1\}$, we compute the utility of optimal differential privacy mechanism as well as its privacy against optimal attack.
  \item For any value of $d_m$ in $\{0.5, 1, \cdots, d_m^{max}\}$, we compute the utility of optimal distortion privacy mechanism as well as its privacy against optimal attack. $d_m^{max}$ is dependent on $\pi$ and is the maximum value that the threshold can take (beyond which there is no solution to the optimization problem). 
  \item For any value of $\epsilon_m$ in $\{0.2, 0.4, \cdots, 1\}$, and for any value of $d_m$ in $\{0.5, 1, \cdots, d_m^{max}\}$, we compute the utility and privacy of the optimal joint mechanism.
\end{enumerate}

Figure~\ref{fig:joint_vs_distdiff} shows the results. By an experiment we refer to the comparison of privacy (or utility) of a joint mechanism (with bounds $\epsilon_m$, $d_m$) with the corresponding differential privacy mechanism (with bound $\epsilon_m$) and the corresponding distortion privacy mechanism (with bound $d_m$). Note that here the thresholds $\epsilon_m$ and $d_m$ are chosen independently as opposed to scenarios 1 (and also 2). We put the results of all the experiments next to each other in the x-axis. Therefore, any vertical cut on the Figure~\ref{fig:joint_vs_distdiff}'s plots contain three points for privacy/utility of $p^*_{\epsilon_m, d_m}$, $p^*_{\epsilon_m}$, and $p^*_{d_m}$. To better visualize the results, we have sorted all the experiments based on the privacy/utility of the joint mechanism.

As the results show, the privacy achieved by the optimal joint mechanism is equal to the maximum privacy that each of the individual differential/distortion mechanisms provides separately. This means that the user would indeed benefit from including a distortion privacy constraint based on her prior leakage into the design criteria of the optimal obfuscation mechanism. This comes at no extra utility cost for the user, as the utility graph shows. In fact, the utility cost of an optimal joint mechanism is not additive and instead is the maximum of the two components, which is the differential privacy mechanism in all tested experiments. The reason behind this is that the differential privacy component makes the joint obfuscation mechanism robust to the case where the background knowledge of the adversary includes not only the prior leakage but also other auxiliary information available to him. 

\begin{figure}[h]
    \centering
    \includegraphics[width=0.98\columnwidth]{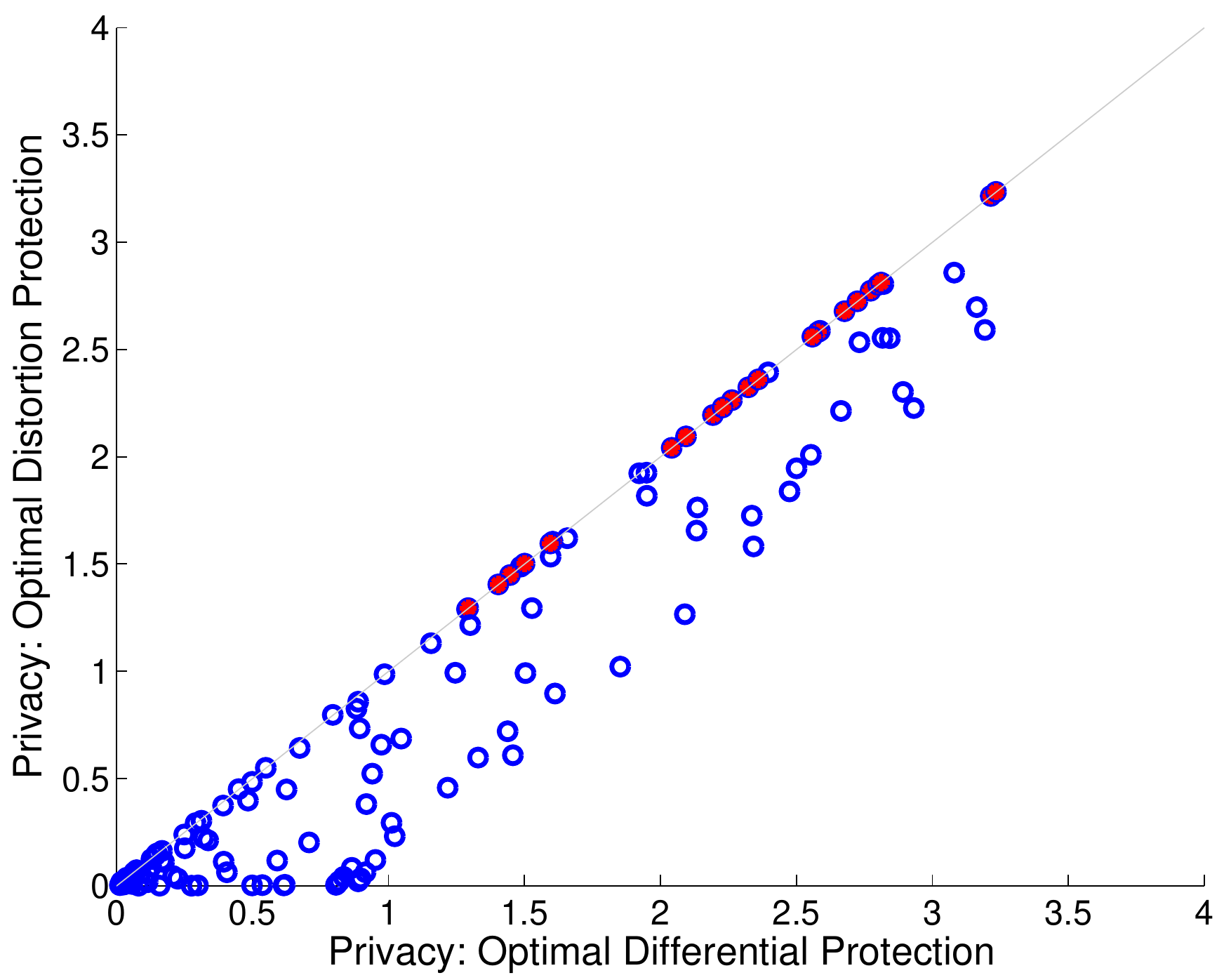}
    \caption{Users' privacy against the optimal attack using optimal differential protection versus using optimal distortion protection. Each circle represents privacy of a user for a different $\epsilon_m$ and for a different prior assumed in the attack. The red dots correspond to the cases where the probability $\pi$ assumed in designing the protection mechanism is the same as the attacker's knowledge.}
    \label{fig:privacy_prior_diff_vs_bayes_optimalattack}
\label{fig:prior}
\end{figure}

\begin{figure*}[t!]
   \centering
   \includegraphics[width=0.98\columnwidth]{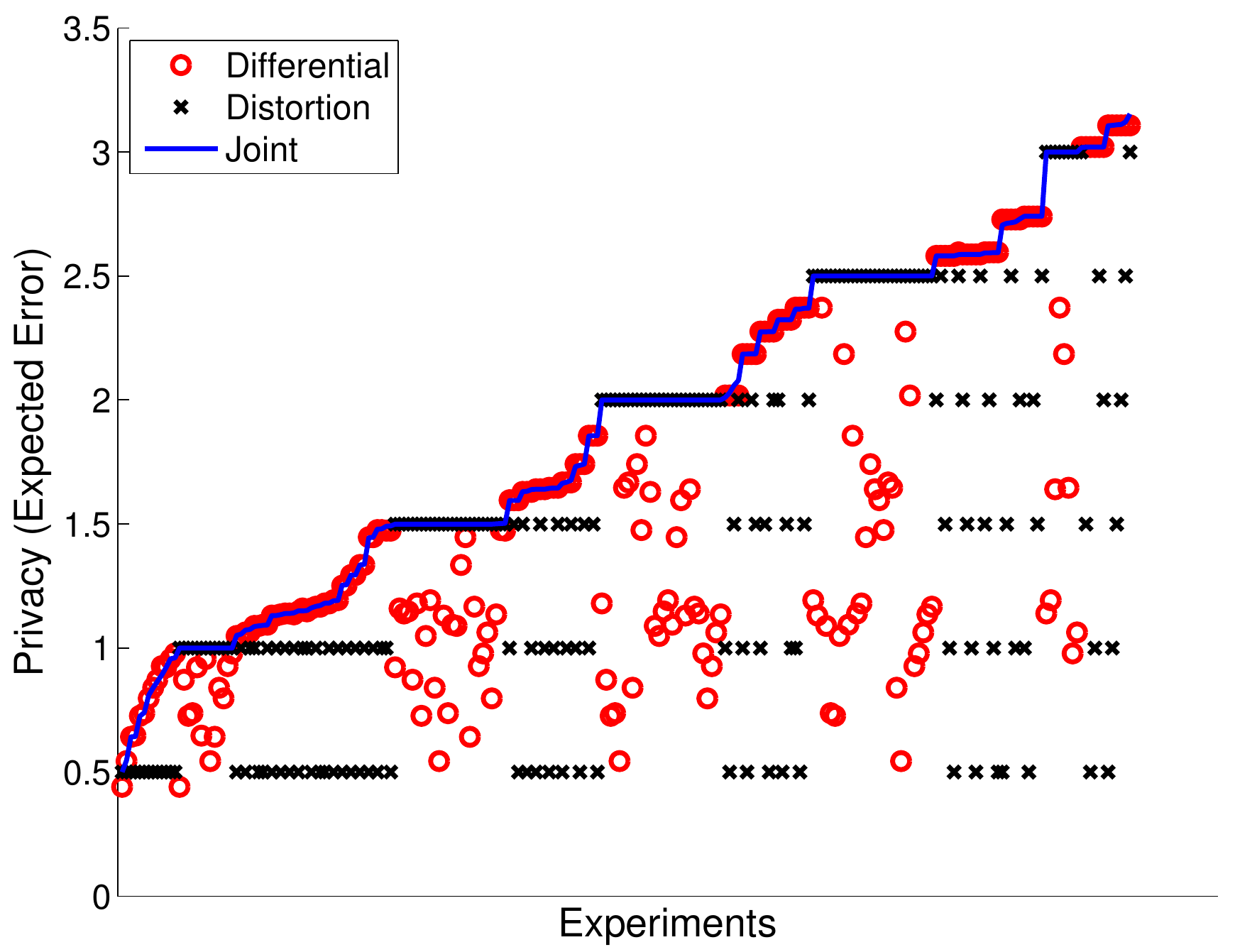}~~~~
   \includegraphics[width=0.98\columnwidth]{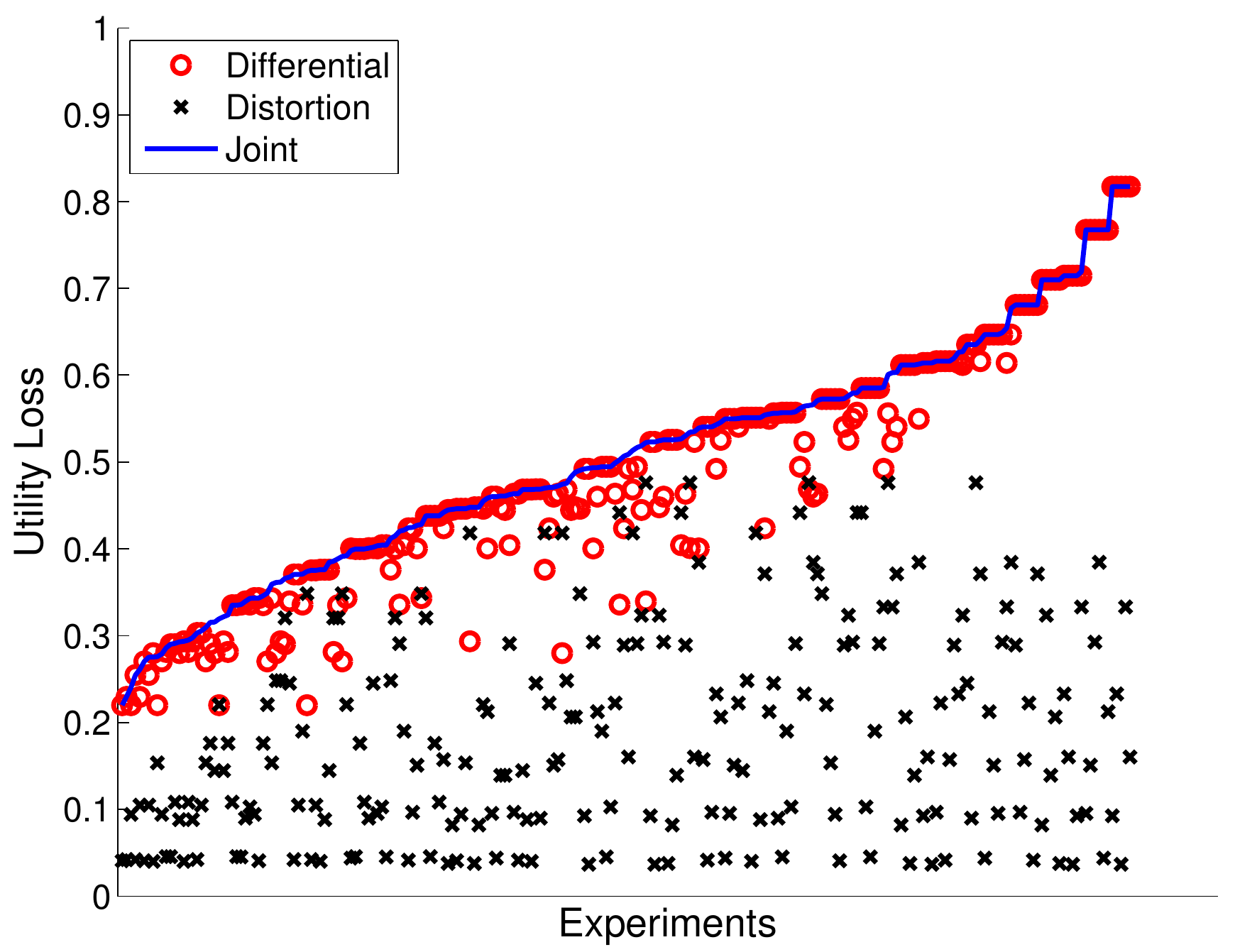}
\caption{The comparison between privacy and utility of optimal joint mechanism $p^*_{\epsilon_m, d_m}$ with the individual protection mechanisms $p^*_{\epsilon_m}$ and $p^*_{d_m}$, i.e., the mechanisms whose bounds are jointly respected in the optimal joint mechanism. The three points on each vertical line represent the results of one such comparison experiment for different values of pairs of privacy thresholds $(\epsilon_m, d_m)$ in $\{0.2, 0.4, \cdots, 1\} \times \{0.5, 1, \cdots, d_m^{max}\}$. }
\label{fig:joint_vs_distdiff}
\end{figure*}

\subsection{Evaluating the Effects of Prior}\label{sec:prior}

When using distortion metric in protecting privacy, we achieve optimal privacy given the user's estimated prior leakage modeled by probability distribution $\pi$ over the secrets. In the optimal attack against various protection mechanisms, a real adversary makes use of a prior distribution over the secrets. In this subsection, we evaluate to what extent a more informed adversary can harm privacy of users further than what is promised by the optimal protection mechanisms. Note that no matter what protection mechanism is used by the user, a more knowledgable adversary will learn more about the secret. In this section, our goal is not to show this obvious fact, but to evaluate how robust our mechanisms are with respect to adversaries with different knowledge accuracy levels.

To perform this analysis, we consider a scenario in which the adversary's assumption on $\pi$, for each user, has a lower level of uncertainty compared to $\pi$. This can happen in the real world when an adversary obtains new evidence about a user's secret that is not used by user for computing $\pi$. Let $\hat{\pi}$ be the other version of $\pi$ assumed by adversary, for a given user. For the sake of our analysis, we generate $\hat{\pi}$ by providing the adversary with more evidence about most frequently visited locations, e.g., home and work. This is equivalent to the scenario in which the adversary knows the user's significant locations, e.g., where the user lives and works. The entropy of $\hat{\pi}$ is less than that of $\pi$, hence it contains more information about the user's mobility.

We construct the protection mechanisms assuming $\pi$, and we attack them by optimal inference attacks, but assuming the lower entropy $\hat{\pi}$ priors. Figure~\ref{fig:prior} illustrates privacy of users for different assumptions of $\hat{\pi}$, using optimal differential protection versus optimal distortion protection (assuming $\pi$). We observe that a more informed adversary has a lower expected error. However, it further shows that an optimal differential protection mechanism compared to an optimal distortion mechanism is more robust to knowledgable adversaries. Note that we set $d_m$ to $AP(p^*_{\epsilon_m})$, according to scenario 1 in Section~\ref{sec:analysis:firsteval}. So, when $\hat{\pi} = \pi$, both optimal protection mechanisms guarantee the same level of privacy. However, as there is more information in $\hat{\pi}$ than in $\pi$, more information can be inferred from the optimal distortion mechanism compared to the differential mechanism.

\section{Conclusions}\label{sec:conclusion}

We have solved the problem of designing {\em optimal} user-centric obfuscation mechanisms for data sharing systems. We have proposed a novel methodology for designing such mechanisms against any {\em adaptive} inference attack, while maximizing users' utility. We have proposed a generic framework for quantitative privacy and utility, using which we formalize the problems of maximizing users' utility under a lower-bound constraint on their privacy. The major novelty of the paper is to solve these optimization problems for both state-of-the-art distortion and differential privacy metrics, for the generic case of any distance function between the secrets. Being generic with respect to the distance functions, enables us to formalize any sensitivity function on any type of secrets. We have also proposed a new privacy notion, joint distortion-differential privacy, and constructed its optimal mechanism that has the strengths of both metrics. We have provided linear program solutions for our optimization problems that provably achieve minimum utility loss under those privacy bounds. 

\section*{Acknowledgements}

We would like to thank the PC reviewers for their constructive feedback, and Kostas Chatzikokolakis for very useful discussions on this work.

\appendix

\section{Optimal Inference Attacks}\label{sec:optimalattack}

Given the user's protection mechanism $p^*$, the inference attack \eqref{eq:lp:adversary:bayesian} is a valid strategy for the adversary, as there is no dependency between the defender and attacker strategies in the case of differential privacy metric.

However, as the differential privacy metric (used in the protection mechanism) does not include any probability distribution on secrets, we can design an inference attack whose objective is to minimize the conditional expected error $E_s$:
\begin{align}\label{eq:error:conditional}
    E_s = \sum_{o, \hat{s}} p^*(o|s) \cdot q(\hat{s}|o) \cdot d(\hat{s}, s)
\end{align}
for all secrets $s$. This is a multi-objective optimization problem \cite{MarlerA04} that does not prefer any of the $E_s$ (for any secret) to another. Under no such preferences, the objective is to minimize $\sum_s E_s$, using weighted sum method with equal weight for each secret.

Thus, the following linear program constitutes the optimal inference attack, under the mentioned assumptions.
\begin{align}\label{eq:lp:adversary:diff}
    \minimize{q} & \sum_{s, o, \hat{s}} p^*(o|s) \cdot q(\hat{s}|o) \cdot d(\hat{s}, s) 
\end{align}

As all the wights of $E_s$ are positive ($= 1$), the minimum of \eqref{eq:lp:adversary:diff} is Pareto optimal \cite{Zadeh63}. Thus, minimizing \eqref{eq:lp:adversary:diff} is sufficient for Pareto optimality. The optimal point in a multi-objective optimization (as in our case) is Pareto optimal ``if there is no other point that improves at least one objective function without detriment to another function'' \cite{MarlerA04, Pareto06}.

An alternative approach is to use the min-max formulation, and minimize the maximum conditional expected error $E_s$ over all secrets $s$. For this, we introduce a new unknown parameter $y$ (that will be the maximum $E_s$). The following linear program solves the optimal inference attack using the min-max formulation. This also provides a necessary condition for the Pareto optimality \cite{Miettinen99}.
\begin{subequations}\label{eq:lp:adversary:diff:minmax}
\begin{align}
    \minimize{q} & y \\
    \st & \sum_{o,\hat{s}} p^*(o|s) \cdot q(\hat{s}|o) \cdot d(\hat{s}, s) \leq y \ , \forall s 
\end{align}
\end{subequations}

We can also consider the expected error conditioned on both secret $s$ and estimate $\hat{s}$ as the adversary's objective to minimize. So, we can use $E_{\hat{s},s} = \pr{\hat{s}|s} \cdot d(\hat{s},s)$ instead of $\sum_{\hat{s}} \pr{\hat{s}|s} \cdot d(\hat{s},s)$ in \eqref{eq:error:conditional}, and use the same approach as in \eqref{eq:lp:adversary:diff:minmax}. The following linear program finds the optimal inference attack that minimizes the conditional expected estimation error over all $s$ and $\hat{s}$, using the min-max formulation.
\begin{subequations}\label{eq:lp:adversary:diff:minmax2}
\begin{align}
    \minimize{q} & y \\
    \st & \sum_o p^*(o|s) \cdot q(\hat{s}|o) \cdot d(\hat{s}, s) \leq y \ , \forall s, \hat{s} 
\end{align}
\end{subequations}

Overall, we prefer the linear program \eqref{eq:lp:adversary:diff} as it has the least number of constraints among the above three. We can also use \eqref{eq:lp:adversary:bayesian} for comparison of optimal protection mechanisms based on distortion and differential metrics.

\section{Proof of Theorem 1}\label{sec:proof}

We construct \eqref{eq:lp:user:utility-privacy:game} from \eqref{eq:lp:user:utility-privacy:nested}. In \eqref{eq:lp:user:utility-privacy:nested}, we condition the optimal obfuscation $p^*$ on its corresponding optimal inference (best response) attack $q^*$. So, for any observable $o$, the inference strategy $q^*(.|o)$ is the one that, by definition of the best response, minimizes the expected error
\begin{align}\label{eq:proof:expectederror}
    \sum_{\hat{s}} q(\hat{s}|o) \sum_s \pi(s) \cdot p(o|s) \cdot d(\hat{s}, s)
\end{align}
Thus, the privacy value \eqref{eq:lp:user:utility-privacy:nested:c1} to be guaranteed is
\begin{align} \label{eq:proof:minq}
    & \sum_{s, o, \hat{s}} \pi(s) \cdot p(o|s) \cdot q^*(\hat{s}|o) \cdot d(\hat{s}, s) = \nonumber\\
    & = \sum_o \minimize{q(.|o)} \sum_{\hat{s}} q(\hat{s}|o) \sum_{s} \pi(s) \cdot p(o|s) \cdot d(\hat{s}, s)
\end{align}

Note that \eqref{eq:proof:expectederror} is an average of $\sum_s \pi(s) \cdot p(o|s) \cdot d(\hat{s}, s)$ over $\hat{s}$, and thus it must be larger or equal to the smallest value of it for a particular $\hat{s}$.
\begin{align}\label{eq:proof:ineq1}
    & \minimize{q(.|o)} \sum_{\hat{s}} q(\hat{s}|o) \sum_{s} \pi(s) \cdot p(o|s) \cdot d(\hat{s}, s) \nonumber\\
    & \ge \minimize{\hat{s}} \sum_s \pi(s) \cdot p(o|s) \cdot d(\hat{s}, s)
\end{align}

Let $q'(.|o)$ be a conditional probability distribution function such that for any given observable $o$,
\begin{align}\label{eq:proof:q}
q'(s'|o) = \left\{
  \begin{array}{l l}
    1 & \text{if $s' = \argmin{\hat{s}} \sum_s \pi(s) \cdot p(o|s) \cdot d(\hat{s}, s)$}\\
    0 & \text{otherwise}
  \end{array} \right.
\end{align}

Note that $q' \in Q$ is a pure strategy that represents one particular inference attack. Moreover, \eqref{eq:proof:minq} constructs $q^*$ such that it optimizes \eqref{eq:proof:expectederror} over the set of all mixed strategies $Q$ that include all the pure strategies. The minimum value for the optimization over the set of all mixed strategies is clearly less than or equal to the minimum value for the optimization over its subset (the pure strategies). Thus, the following inequality holds.
\begin{align}\label{eq:proof:ineq2}
    & \sum_{\hat{s}} q^*(\hat{s}|o) \sum_{s} \pi(s) \cdot p(o|s) \cdot d(\hat{s}, s) \nonumber\\
    & = \minimize{q(.|o)} \sum_{\hat{s}} q(\hat{s}|o) \sum_{s} \pi(s) \cdot p(o|s) \cdot d(\hat{s}, s) \nonumber\\
    & \leq \sum_{\hat{s}} q'(\hat{s}|o) \sum_{s} \pi(s) \cdot p(o|s) \cdot d(\hat{s}, s) \nonumber\\
    & = \minimize{\hat{s}} \sum_s \pi(s) \cdot p(o|s) \cdot d(\hat{s}, s)
\end{align}

Therefore, from inequalities \eqref{eq:proof:ineq1} and \eqref{eq:proof:ineq2} we have
\begin{align}
& \sum_{s, o, \hat{s}} \pi(s) \cdot p(o|s) \cdot q^*(\hat{s}|o) \cdot d(\hat{s}, s) \nonumber\\
& = \sum_o \minimize{\hat{s}} \sum_s \pi(s) \cdot p(o|s) \cdot d(\hat{s}, s) 
= \sum_o x(o)
\end{align}
where $x(o) = \minimize{\hat{s}} \sum_s \pi(s) \cdot p(o|s) \cdot d(\hat{s}, s)$, or equivalently $x(o) \leq \sum_s \pi(s) \cdot p(o|s) \cdot d(\hat{s}, s), \forall \hat{s}$.

Thus, the constraint \eqref{eq:lp:user:utility-privacy:nested:c1} in the linear program \eqref{eq:lp:user:utility-privacy:nested} is equivalent to (and can be replaced by) the constraints \eqref{eq:lp:user:utility-privacy:game:c1} and \eqref{eq:lp:user:utility-privacy:game:c2} in the linear program \eqref{eq:lp:user:utility-privacy:game}.

\section{Approximating The Optimal Mechanisms}\label{sec:approx}

\begin{figure*}[t!]
    \centering
    \includegraphics[width=0.98\columnwidth]{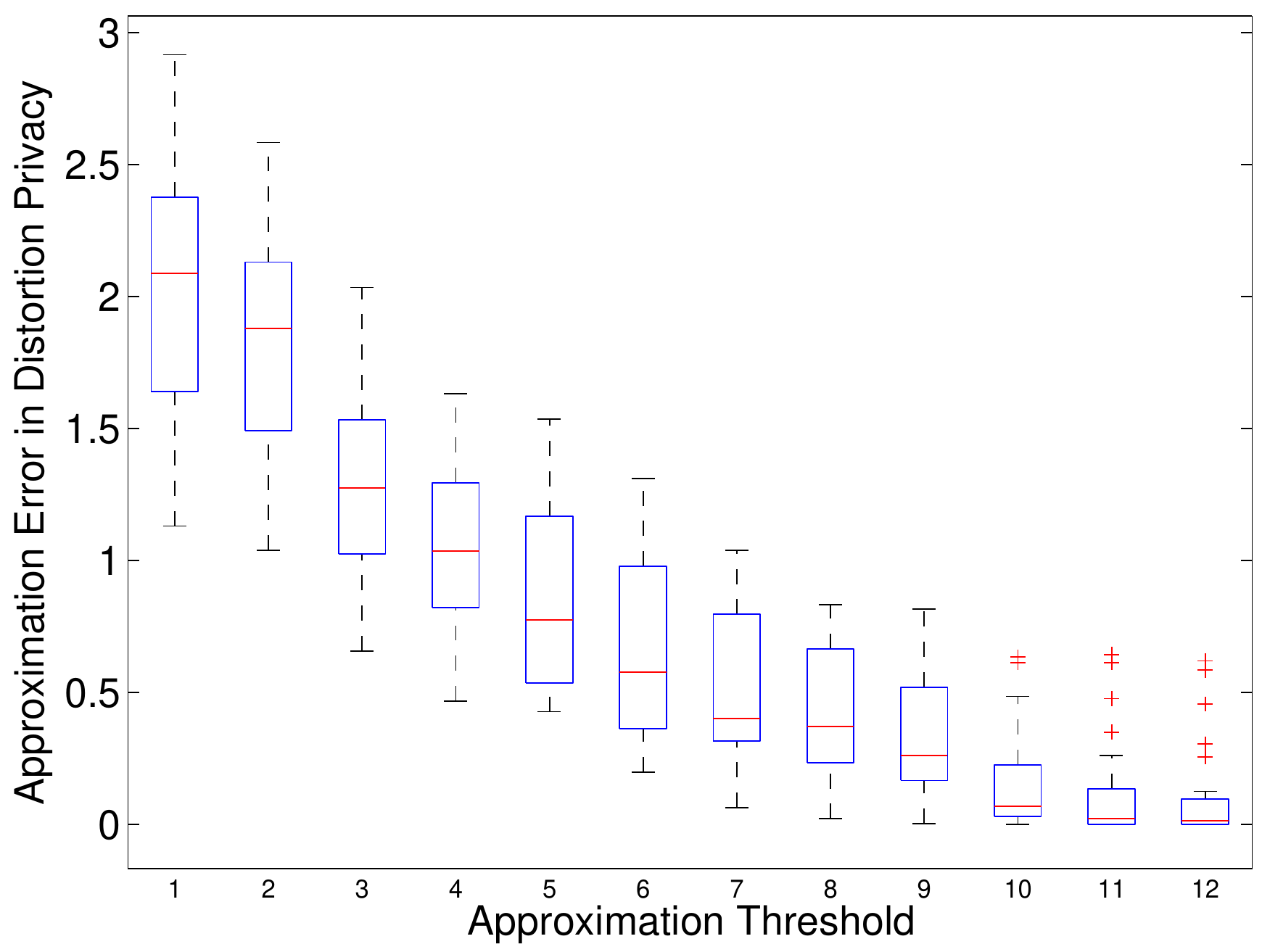}~~~~
    \includegraphics[width=0.98\columnwidth]{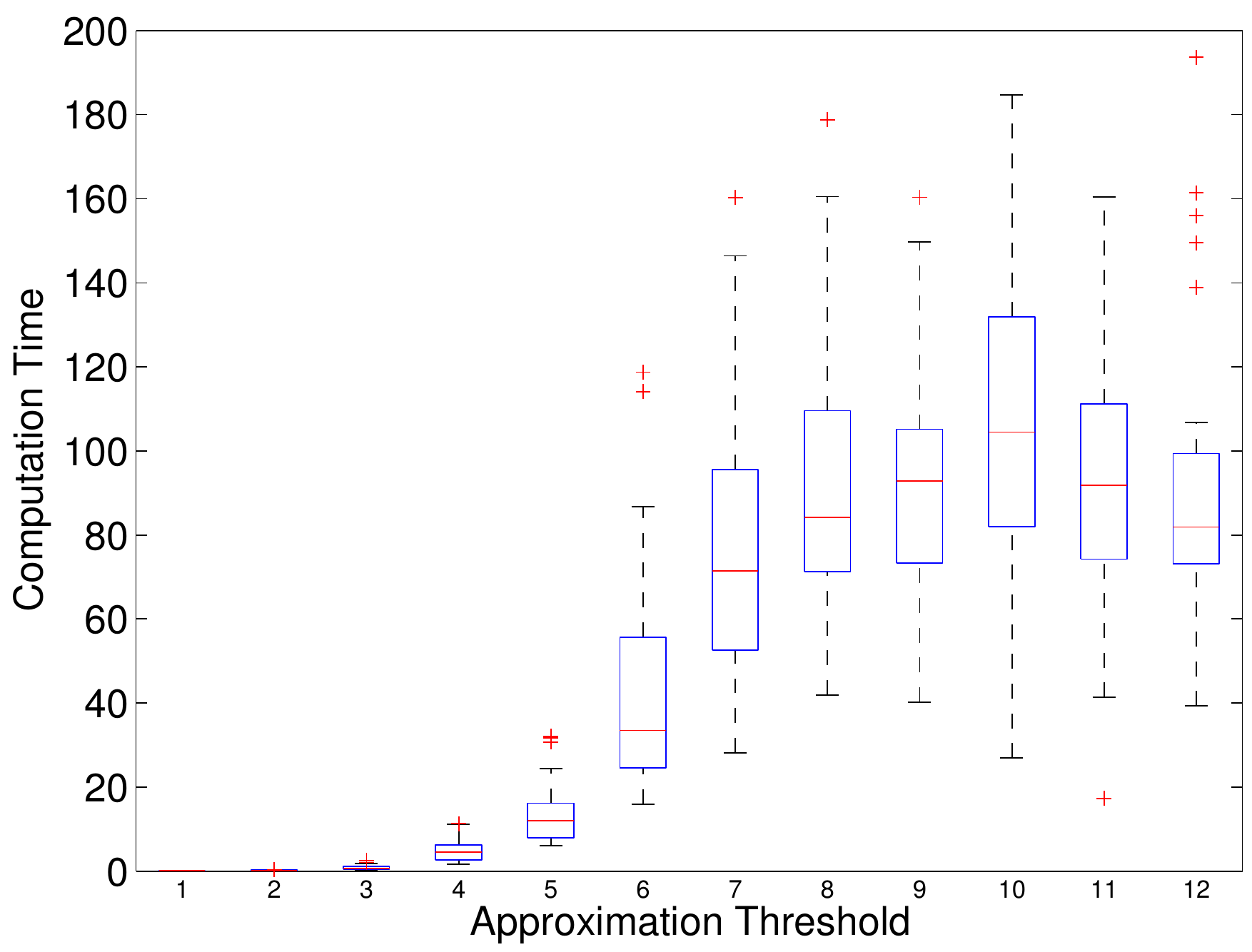}
    \caption{Error (in km) and computation time (in sec) of approximate optimal joint distortion-differential privacy protection mechanisms, by considering a subset of optimization constraints. We  consider only the constraints for which the distance $d()$ (in km) between observation and secret and the distance between two secrets is less than the x-axis. In the left-hand side figure, the y-axis shows the distribution of the difference between privacy of users with and without approximation. In the right-hand side figure, the y-axis represents the total computation time of solving the linear program of the approximate optimal joint mechanism in Matlab on a machine with 4 core CPU model Intel(R) Xeon(R) 2.40GHz. The central mark in each box shows the median value, computed over all users. The boxplot also shows the $25^{th}$, and $75^{th}$ percentiles as well as the outliers.}
    \label{fig:privacy_appx}
\end{figure*}

Here, we briefly discuss the computational aspects of the design of optimal protection mechanisms. Although the solution to linear programs provides us with the optimal protection mechanism, their computation cost is quadratic (for distortion mechanisms) and cubic (for differential mechanisms) in the cardinality of the set of secrets and observables. Providing privacy for a large set of secrets needs a high computation budget. To establish a balance between the computation budget and privacy requirements, we can make use of approximation techniques to design optimal protection mechanisms. We explore some possible approaches.

Linear programming \cite{BoydV04} is one of the fundamental areas of mathematics and computer science, and there is a variety of algorithms to solve a linear program. Surveying those algorithms and evaluating their efficiencies is out of the scope of this paper. These algorithms search the set of feasible solutions of a problem for finding the optimal solution that meets the constraints. Many of these algorithms are iterative and they converge to the optimal solution as the number of iterations increases \cite{NesterovN93, GrotschelMLS81}. Thus, a simple approximation method is to stop the iterative algorithm when our computation budget is over. Other approximation methods exist. For example, \cite{FariasV06} suggests a sampling algorithm to select a subset of constraints in an optimization problem to speed up the computation. Moreover, we can rely on the particular structure of secrets to reduce the set of constraints \cite{BordenabeCP14}.

We can implement those approximation techniques to solve approximately optimal protection mechanisms in an affordable time. Furthermore, we can rely on the definition of privacy to find the constraints that have a minor contribution to the design of the protection mechanism. In this section, we study one approximation method, following the intuition behind the differential privacy bound: we remove the constraints for which the distance $d(s, s')$ is larger than a threshold. We can justify this by observing that, in the definition of differential privacy metric \eqref{eq:privacy:diff:gneric:mult}, the privacy is more protected when for secrets $s, s'$, the distance $d(s, s')$ is small. To put this in perspective, note that if we use the original definition of differential privacy, there would not be any constraint if $d(s,s') > 1$. We also apply this approximation to the distance between observables and secrets.

In Figure~\ref{fig:privacy_appx}, we show the privacy loss of users as well as the speed-up of their computation due to approximation. We performed the computation on a machine with 4 core CPU model Intel(R) Xeon(R) 2.40GHz. As we increase the approximation threshold (which is the distance beyond it we ignore the constraints), the approximation error goes to zero. This suggests that, for a large set of secrets, if we choose a relatively small threshold the approximated protection mechanism provides almost the same privacy level as in the optimal solution. The computation time, however, increases as the approximation error decreases (due to increasing the approximation threshold). Figure~\ref{fig:privacy_appx} captures such a tradeoff of our approximation method. 

\end{document}